\documentclass[11pt]{article}

\usepackage{authblk}
\author{Daniil Bargman\thanks{daniil.bargman.22@ucl.ac.uk}}
\affil{Institute of Finance and Technology, University College London}
\usepackage{geometry}
\usepackage[authoryear,round]{natbib}

\let\citep\citet
\usepackage[breaklinks]{hyperref}
\hypersetup{colorlinks=true,urlcolor=blue,linkcolor=blue,citecolor=blue}
\usepackage{afterpage}
\usepackage{ulem}
\usepackage{enumitem}
\renewlist{itemize}{itemize}{9}
\setlist[itemize]{label=$\circ$}
\usepackage{tabularx}
\usepackage{multirow}
\usepackage{adjustbox}
\usepackage{threeparttable}
\usepackage{caption}
\usepackage{booktabs}
\usepackage{arydshln}
\usepackage{array}
\usepackage{makecell}

\usepackage{floatrow}
\usepackage{nccmath}
\usepackage{amsthm}
\usepackage{amsfonts}
\usepackage{parskip}
\usepackage{cleveref}
\usepackage{stackengine}
\usepackage{pdflscape}
\usepackage{bm}
\usepackage{cases}
\usepackage{mathtools}
\usepackage{placeins}
\newtheorem{definition}{Definition}[section]

\newtheorem{proposition}{Proposition}[section]

\newcolumntype{L}{>{$}l<{$}}
\newcolumntype{C}{>{$}c<{$}}
\newcolumntype{P}[1]{>{\centering\arraybackslash}p{#1}}
\def\a{\mathbf{\bm{a}}}
\def\ta{\tilde{a}}
\def\y{\mathbf{\bm{y}}}
\def\ty{\tilde{y}}
\def\yh{\hat{\ty}}
\def\tmy{\tilde{\y}}
\def\myh{\hat{\tmy}}
\def\Y{\mathbf{\bm{Y}}}

\def\O{\mathbf{\bm{\Omega}}}
\def\k{\mathbf{\kappa}}

\def\x{\mathbf{\bm{x}}}
\def\P{\mathbf{\bm{L}}}
\def\B{\mathbf{\bm{B}}}

\def\vbdsum{\vb^{\oplus}}
\def\vbdsumt{\left( \vbdsum \right)'}

\def\tx{\tilde{x}}
\def\tX{\tilde{X}}

\def\A{\mathbf{\bm{A}}}
\def\Adsum{\A^{\oplus}}
\def\Adsumt{\left( \Adsum \right)'}
\def\adsum{\a^{\oplus}}
\def\adsumt{\left( \adsum \right)'}

\def\X{\mathbf{\bm{X}}}
\def\tX{\widetilde{X}}

\def\1{\mathbf{\bm{1}}}
\def\t1od{\1_{\o}^{\oplus}}
\def\0{\mathbf{\bm{0}}}
\def\F{F}

\def\Sm{\mathbf{\bm{\Sigma}}}
\def\M{\mathbf{\bm{M}}}

\def\V{\bm{\mathbf{\Theta}}}

\def\Smd{\Sm^{d}}
\def\g{\bm{\gamma}}
\def\gh{\bm{\hat{\gamma}}}

\def\ch{\hat{c}}
\def\w{\bm{\mathbf{w}}}
\def\wh{\bm{\mathbf{\hat{w}}}}
\def\o{\bm{\omega}}
\def\oh{\bm{\hat{\omega}}}
\def\od{\o^{\oplus}}
\def\odt{\left( \od \right)'}

\def\b{\bm{\beta}}
\def\bh{\bm{\hat{\beta}}}
\def\vb{\bm{b}}
\def\p{\bm{\phi}}
\def\phih{\hat{\phi}}
\def\ph{\bm{\hat{\phi}}}
\def\pw{\brac{\p \otimes \w}}
\def\pi{\brac{\p \otimes I_n}}
\def\piw{ \pi \w}
\def\wi{\brac{I_{V_a} \otimes \w}}

\def\wip{\wi \p}
\def\bo{\brac{\b \odot \o}}
\def\bi{\brac{\b \odot I_{\o}}}
\def\bio{\bi \o}
\def\oi{\brac{I_{\b} \odot \o}}
\def\oib{\oi \b}
\def\bko{\brac{\b \otimes \o}}
\def\bki{\brac{\b \otimes I_m}}

\def\oki{\brac{I_F \otimes \o}}

\def\uo{\brac{\u \odot \o}}
\def\ui{\brac{\u \odot I_{\o}}}

\def\oh{\bm{\hat{\omega}}}

\def\v{\vartheta}
\def\tv{\bm{\mathbf{v}}}
\def\s2{\sigma^2}
\def\ts2{\bm{\v}}
\def\l{\lambda}
\def\ls2{\l_{\s2}}
\def\ly{\l_{y}}

\def\r{\rho}
\def\rs2{\r_{\s2}}
\def\ry{\r_{y}}

\def\tl{\bm{\mathbf{\lambda}}}
\def\lp{\bm{l}_p}

\def\u{\bm{\mathbf{u}}}

\def\tg{\tilde{g}}
\def\tr{\tilde{r}}

\setstackgap{L}{8pt}

\newcommand{\Lagr}{\mathcal{L}}

\newcommand{\brac}[1]{\left( {#1} \right)}
\newcommand{\seq}[1]{\left< {#1} \right>}
\newcommand*\mean[1]{\overline{#1}}
\DeclareMathOperator{\diag}{diag}
\allowdisplaybreaks
\date{\today}
\title{Latent Variable Modelling by Supervised Diffusion}
\begin{document}

\maketitle
\begin{abstract}

This paper proposes a new methodological framework for estimating
inferential models with latent variables. It also introduces a new
latent variable regression model called LARX: an extension of the
ubiquitous autoregressive model with exogenous inputs (ARX) in which any
or all input variables can be latent. In deriving the LARX model, a
minor contribution is also made to the field of matrix calculus: A new
matrix operator is defined and applied to solve a class of Lagrangian
optimisation problems with interactions between multiple coefficient
vectors subject to case-by-case constraints.

In the empirical section, the LARX model is used to re-examine the
relationship between stock market performance and real economic activity
in the United States. The LARX model attains an out-of-sample R-squared
of up to 79.7\% compared to 50.3\% for the baseline OLS approach. It also
reveals new information about the underlying drivers of the relationship
between stock returns and economic growth, including the predictive
power of sector rotations.

\end{abstract}
\section{Introduction}
\label{introduction}
A latent variable (LV) is an unobserved process which influences the
behaviour of an observed system. Many inferential models in economics
and finance use latent variables, which adds a layer of complexity to
their empirical estimation process. Prominent examples of such models
include the business cycle theory (\citep{Burns-Mitchell-1946}) and the
asset pricing theory (APT) (\citep{Ross-1976}), both of which are
characterised by a small number of unobserved processes driving a much
larger number of observed variables\footnote{The business cycle theory states that the unobserved state of the
business cycle drives a large number of observed macroeconomic
aggregates such as growth, inflation and unemployment. In APT, long-term
asset returns are determined by the respective assets' exposures to
unobserved risk premia.}.

One of the most popular methodologies for approximating latent variables
in both descriptive and inferential settings is principal component
analysis (PCA) attributed to \citep{Pearson-1901} and
\citep{Hotelling-1933}. PCA uses a group of observed variables to
construct a sequence of latent variables, known as Principal Components
(PC) or Diffusion Indices (DI)\footnote{Principal components are referred to as ``diffusion indices'' in
\citep{Stock-2002}, which is one of the first papers to demonstrate the
superior forecasting power of latent macroeconomic factors approximated
by principal components with respect to a wide range of US economic
aggregates. The term is not to be confused with the colloquial meaning
of a diffusion index as a summary statistic describing the share of
positive and negative values in a set of survey responses or
macroeconomic indicators, such as \href{https://www.ismworld.org/supply-management-news-and-reports/reports/ism-pmi-reports/}{the suite of US business activity
indices produced by the Institute of Supply Management}. This paper
adopts the terminology of \citep{Stock-2002} and uses the term
``diffusion index'', or DI, to refer to a linear combination of observed
variables constructed with the goal of approximating a latent variable
or factor.}, each capturing one independent
direction of variance. DIs have useful statistical properties, including
their ability to reduce the dimensionality of the input dataset and to
eliminate multicollinearity in the final regression problem. However,
PCA also has two major drawbacks: interpretability, and the restrictive
assumptions required to ensure the optimality of DI-based factor models
in inferential settings (e.g., see \citep{Bai-2006}).

This paper builds on the conceptual foundations of PCA and its
multivariate cousin, Canonical Correlation Analysis (CCA)
(\citep{Hotelling-1936}), to propose a new methodological framework for
estimating inferential models with latent variables. The new framework
is given the name `'Supervised Diffusion'' (SDF), following the
terminology of \citep{Stock-2002} which refers to principal components
as diffusion indices. SDF capitalises on the idea of latent variables as
linear combinations of observed variables. However, contrary to PCA, the
linear combination weights are optimised for inferential purposes.

The SDF formula makes it possible to transform any traditional
regression methodology into a latent variable model. As a practical
example, the ubiquitous autoregressive model with exogenous inputs (ARX)
is used to derive a new regression model called LARX: a superset of the
ARX model in which any or all variables can be latent. Like ARX, the
LARX model presents with several interesting special cases, including a
parsimonious lead-lag regression model referred to as ``Latent Shock
Regression'' (LSR), as well as a new canonical decomposition technique
best described as Canonical Autocorrelation Analysis (CAA).

In the broader ecosystem of latent variable research, SDF and the LARX
model contribute to a body of work extending beyond economics and
finance, where models based on CCA and its close relative Partial Least
Squares (PLS) (\citep{Wold-1982,Wold-1975}) are studied under the term
Latent Variable Regression (LVR). Over the past 90 years, LVR models
have been developed and applied in fields like industrial chemistry
(e.g., \citep{Burnham-MacGregor-1996,Burnham-McGregor-1999}), machine
learning (e.g., \citep{Wang-2020,Dai-2020,Vaerenberg-2018,Chi-2013}), and medicine, among other quantitative disciplines
(\citep{Uurtio-2017} provides a good overview). Traditional use cases
for these models include dimensionality reduction, identification of the
directions of correlation in multivariate data streams (e.g.,
\citep{Burnham-MacGregor-1996,Dong-2018,Qin-2021}), and multi-label
classification of images, videos, audio, and hand-written text (e.g.,
\citep{Wang-2020,Dai-2020,Vaerenberg-2018,Chi-2013}).

The rest of this paper is organised as follows. Section
\ref{blockwise_direct_sum} defines a new matrix operator which is
subsequently used to derive the LARX model. Section \ref{sdf} introduces the
supervised diffusion framework and examines its relationship to PCA,
least-squares regression and portfolio optimisation. Section \ref{larx}
derives the LARX methodology, making a minor contribution to the field
of matrix calculus. Section \ref{larx_special_cases} examines a few special
cases of the LARX model, including CCA, LSR and CAA. Section
\ref{empirical_results} applies the LARX model in an empirical setting by
examining the predictive power of the US stock market with respect to
real US economic activity. Section \ref{concluding_remarks} concludes and
discusses possible avenues for future work.
\section{Block-wise Direct Sum Operator}
\label{blockwise_direct_sum}
Some of the mathematical derivations in this paper involve block-wise
operations between matrices and vectors with blocks along one dimension.
These derivations become much easier with the help of a shorthand
notation which we can call a block-wise direct sum \(\Adsum\) for an
arbitrary matrix or vector \(\A\). This section briefly introduces the
block-wise direct sum operator and its relevant properties.

Let \(\seq{\A_i | 1 \leq i \leq k}\) be a sequence of \(k\) real
matrices with arbitrary dimensions. If all \(\A_i\) have the same number
of columns, they can be concatenated vertically into a matrix with row
blocks. If all \(\A_i\) have the same number of rows, they can be
concatenated horizontally into a matrix with column blocks. Formally:

\begin{flalign*}
& \seq{\A} \equiv \seq{\A_i} := \seq{\A_i | 1 \leq i \leq k}
  & \text{ a sequence of matrices of arbitrary dimensions} \\[10pt]
& \A \equiv \A_v := \left[ \seq{\A_i} \right]_v 
  & \text{ a vertical concatenation of } \seq{\A_i} \text{ (default)} \\[10pt]
& \A_h := \left[ \seq{\A_i} \right]_h 
  & \text{ a horizontal concatenation of } \seq{\A_i}
\end{flalign*}

The block-wise direct sum operator \(\Adsum\) is then defined as:

\begin{flalign*}
& \Adsum \equiv \seq{\A}^{\oplus} \equiv \seq{\A_i}^{\oplus}
         := \A_1 \oplus \A_2 \oplus \A_3 \oplus \cdots \oplus \A_k
    = \begin{pmatrix}
        \A_1, & \0, & \cdots & \0 \\
        \0, & \A_2, & \cdots & \0 \\
        \vdots & \vdots & \ddots & \vdots \\
        \0, & \0, & \cdots & \A_k
      \end{pmatrix} &
\end{flalign*}

Compatibility of dimensions is not a constraint for a direct sum between
matrices, so a block-wise direct sum is defined for any sequence of
matrices \(\seq{\A_i}\). The result is always a matrix whose number of
rows (columns) is the sum total number of rows (columns) across all
comprising \(\A_i\).

The matrix \(\Adsum\) can itself be mapped to either a sequence of \(k\)
row blocks or a sequence of \(k\) column blocks without slicing through
the original matrices in the sequence. For most use cases the block
structure of \(\Adsum\) will not be relevant. For the remaining
scenarios let us apply the same shorthand notation as above:

\begin{flalign*}
& \Adsum \equiv \Adsum_v
    := \brac{ \begin{array}{cccc}
        \A_1, & \0, & \cdots & \0 \\ \hline
        \0, & \A_2, & \cdots & \0 \\ \hline
        \vdots & \vdots & \ddots & \vdots \\ \hline
        \0, & \0, & \cdots & \A_k
      \end{array} }, \quad 
\Adsum_h
    := \brac{ \begin{array}{c|c|c|c}
        \A_1 & \0 & \cdots & \0 \\
        \0 & \A_2 & \cdots & \0 \\
        \vdots & \vdots & \ddots & \vdots \\
        \0 & \0 & \cdots & \A_k \\
      \end{array} } &
\end{flalign*}

One of the main use cases for the block-wise direct sum operator lies in
defining block-wise inner products and block-wise quadratic forms between
matrices and vectors. For example, let \(\B\) be a matrix with the same
number of rows and the same block structure as \(\A\). The product
\(\Adsumt \B\) then yields:

\begin{flalign*}
& \Adsumt \B
    = \begin{pmatrix}
        \A_1' & \0 & \cdots & \0 \\
        \0 & \A_2' & \cdots & \0 \\
        \vdots & \vdots & \ddots & \vdots \\
        \0 & \0 & \cdots & \A_k'
      \end{pmatrix} \left( \begin{array}{c}
    \B_1 \\ \hline
    \B_2 \\ \hline
    \vdots \\ \hline
    \B_k
 \end{array} \right)
    = \left( \begin{array}{c}
    \A_1' \B_1 \\ \hline
    \A_2' \B_2 \\ \hline
    \vdots \\ \hline
    \A_k' \B_k
 \end{array} \right) &
\end{flalign*}

In our case, this operation will be particularly useful for solving
Lagrangian optimisation problems in the presence of piecemeal
constraints on the coefficient vector, e.g., when one portion of the
coefficient vector must have a unit length constraint, another a
zero-sum constraint, and so on. Several properties of the block-wise
direct sum operator will be relevant in this context. First of all, we
note that the transpose of \(\Adsum\) is the same as the block-wise
direct sum of \(\A'\):

\begin{proposition} \label{prop:bds_vs_transpose} The function
composition of the block-wise direct sum operator and the transpose
operator is commutative. In other words, \( \Adsumt = \brac{ \A'
}^{\oplus}\).
\end{proposition}

\begin{proof}
\begin{flalign*}
& {\Adsumt} = \begin{pmatrix}
        \A_1' & \0 & \cdots & \0 \\
        \0 & \A_2' & \cdots & \0 \\
        \vdots & \vdots & \ddots & \vdots \\
        \0 & \0 & \cdots & \A_k'
      \end{pmatrix}
    = \A_1' \oplus \A_2' \oplus \A_3' \oplus \cdots \oplus \A_k'
    = \brac{\A'}^{\oplus} &
\end{flalign*} \end{proof}

Second, the block-wise direct sum of a vector can be written as a
block-wise Kronecker product as defined in \citep{Khatri-Rao-1968},
which we denote by ``\(\odot\)'':

\begin{proposition} \label{prop:bds_vs_odot} Let \(\a\) comprise \(k\)
blocks given by the sequence of vectors \(\seq{\a_i | 1 \leq i \leq
k}\). The matrix \(\adsum\) can be expressed as a block-wise Kronecker
product between \(\a\) and an identity matrix \(I_k\) with vector
blocks along the same dimension as \(\a\).
\end{proposition}
\begin{proof}For a column vector we have:
\begin{flalign*}
& \adsum = \begin{pmatrix}
        \a_1 & \0 & \cdots & \0 \\
        \0 & \a_2 & \cdots & \0 \\
        \vdots & \vdots & \ddots & \vdots \\
        \0 & \0 & \cdots & \a_k
      \end{pmatrix} = \brac{ \begin{array}{c}
        \a_1 \otimes \begin{bmatrix}1, & 0, & \cdots & 0 \end{bmatrix} \\ \hline
        \a_2 \otimes \begin{bmatrix}0, & 1, & \cdots & 0 \end{bmatrix} \\ \hline
        \vdots \\ \hline
        \a_k \otimes \begin{bmatrix}0, & 0, & \cdots & 1 \end{bmatrix} 
      \end{array} } = \brac{ \begin{array}{c}
        \a_1 \\ \hline
        \a_2 \\ \hline
        \vdots \\ \hline
        \a_k
      \end{array} } \odot \brac{ \begin{array}{cccc}
        1, & 0, & \cdots & 0 \\ \hline
        0, & 1, & \cdots & 0 \\ \hline
        \vdots & \vdots & \ddots & \vdots \\ \hline
        0, & 0, & \cdots & 1
      \end{array} } = \a \odot I_k &
\end{flalign*}
The proof for a row vector follows by symmetry. \end{proof}

Third, for two column vectors \(\a\) and \(\vb\) with the same length
and row block structure, the operation \(\adsumt \vb\) is symmetric:

\begin{proposition} \label{prop:bds_vs_vectors}
For column vectors \(\a\) and \(\vb\) with identically sized row blocks
\(\a_1,\a_2,\a_3,\ldots,\a_k\) and \(\vb_1,\vb_2,\vb_3,\ldots,\vb_k\),
respectively, \(\adsumt \vb = \vbdsumt \a \).
\end{proposition}

\begin{proof}
\begin{flalign*}
& \adsumt \vb = \begin{pmatrix}
        \a_1' & \0 & \cdots & \0 \\
        \0 & \a_2' & \cdots & \0 \\
        \vdots & \vdots & \ddots & \vdots \\
        \0 & \0 & \cdots & \a_k'
      \end{pmatrix} \left( \begin{array}{c}
    \vb_1 \\ \hline
    \vb_2 \\ \hline
    \vdots \\ \hline
    \vb_k
 \end{array} \right)
  = \left( \begin{array}{c}
    \a_1' \vb_1 \\ \hline
    \a_2' \vb_2 \\ \hline
    \vdots \\ \hline
    \a_k' \vb_k
 \end{array} \right)
  = \left( \begin{array}{c}
    \vb_1' \a_1 \\ \hline
    \vb_2' \a_2 \\ \hline
    \vdots \\ \hline
    \vb_k' \a_k
 \end{array} \right) = \vbdsumt \a &
\end{flalign*} \end{proof}

Lastly, the block-wise direct sum operator is commutative with respect to
a certain class of operations over matrix sequences. Specifically, for a
sequence of matrices \(\seq{\A_i | 1 \leq i \leq k}\) and an operation
\(f\) over matrix sequences of length \(k\), it can be shown that \(f
\brac{ \seq{\A} }^{\oplus} = \left[ f \brac{ \seq{\Adsum} } \right]\)
if \(f\) satisfies certain conditions. This, in turn, can be used to
prove that for two vectors \(\a\) and \(\vb\) with \(k\) row blocks
each, the block-wise Kronecker product \(\a \odot \vb\) can be factorised
in the same way as the traditional Kronecker product, namely:

\begin{equation*}
\a \odot \vb = \brac{\a \odot I_{\vb}} \vb = \brac{ I_{\a} \odot \vb } \a
\end{equation*}

where \(I_{\a}\) and \(I_{\vb}\) are identity matrices with the same row
block structure as \(\a\) and \(\vb\), respectively. The corresponding
derivations are deferred to \ref{bds_vs_mmul} and \ref{block_kron_factorisation}.
\section{A Conceptual Case for a Supervised Diffusion Framework}
\label{sdf}
An inferential model with latent variables (from now, simply a ``latent
variable model'', or LVM for short) has two simultaneous objectives: to
estimate a functional relationship, and to approximate the latent
variables in it. Several types of latent variable models are widely used
in economics and finance\footnote{Two popular latent variable models in modern econometrics are
instrumental variable regression (e.g., \citep{Stock-2003-IV}) and
hidden Markov models (\citep{Baum-1966}). In the former, the causal
relationship between the explanatory and the dependent is obscured by an
unobserved common factor. In the latter, the variables and their
relationships are influenced by unobserved changes in regime such as the
stage of the business cycle or a bull/bear market.}, but this paper is concerned with a
special class of LVMs in which unobserved processes are approximated as
linear combinations of observed variables. Principal Component Analysis
(PCA) (\citep{Pearson-1901,Hotelling-1933}) is arguably the best-known
model of this kind.

This section proposes a new methodological framework for estimating
latent variable models. It is referred to as a ``Supervised Diffusion
Framework'' (SDF) following the terminology of \citep{Stock-2002} which
uses the term ``Diffusion Index'' (DI) to describe a composite indicator
obtained by PCA. Conceptually, SDF brings PCA into an inferential
setting: If the goal of PCA is to describe the latent sources of
variation in a single set of variables, the goal of SDF is to describe
the latent source(s) of functional dependency between two or more sets.
In this context, SDF also generalises the method of Canonical
Correlation Analysis (CCA) (\citep{Hotelling-1936}) whose goal is to
identify the latent sources of pairwise correlation -- a special type of
functional relationship (more on this in Section \ref{cca}).

Let us start with a couple of formal definitions:

\begin{definition}\label{def:lv}

A latent variable \(\ty\) is an unobserved process which can be
approximated as a linear combination over a vector of observed random
variables \(Y = \begin{bmatrix} y_1, & y_2, & y_3, & \ldots, & y_n
\end{bmatrix}\) with measurement error \(r_y\), such that \(\ty = Y\w +
r_y\) where \(\w\) has dimensions \(n \times 1\) and can't be a multiple
of a standard basis vector.

\end{definition}

\begin{definition}\label{def:lvm}

A latent variable model (LVM) is a regression model in which at least
one variable is latent.

\end{definition}

Now let us consider a generalised formula for a standard multiple
regression between a dependent variable \(y\) and a sequence of \(\k\)
explanatory variables \(\seq{x_j} := \seq{x_1,x_2,x_3,\ldots,x_{\k}} =
\seq{x_j}_{j=1}^{\k}\):

\begin{equation}\label{eqn:reg}
y = \F_{\g}(x_1,x_2,x_3,\ldots,x_{\k}) + e \equiv \F_{\g}(\seq{x_j}) + e
\end{equation}

Here, \(\g\) denotes a vector of regression parameters, \(e\) is a
mean-zero error term, and \(F_{\g}: \mathbb{R} \longrightarrow
\mathbb{R}\) is generally assumed to be at least once differentiable
with respect to \(\g\). The standard least-squares solution to (\ref{eqn:reg})
is obtained by finding the vector \(\gh\) which minimises the variance
of \(e\) using some sample observation vectors \(\y\) and \(\x_j\)
containing the data for \(y\) and \(x_j\), respectively.

In a slight abuse of notation, let \(\F(\seq{\a_j})\) represent a
row-wise operation over the sequence \(\seq{\a_j}\) when \(\a_j\) are
sample vectors (or matrices). In other words, given a sequence of sample
vectors \(\seq{\a_j}_{j = 1}^{\k}\) where \(\a_j = \begin{bmatrix}
a_{1,j}, & a_{2,j}, & a_{3,j}, & \ldots, & a_{s,j} \end{bmatrix}'\) for
all \(j\) and \(s\) is the sample size, let:

\begin{equation*} 
\F(\seq{\a_j}) = \begin{bmatrix}
\F(\seq{a_{1,j}}_{j = 1}^{\k}), & \F(\seq{a_{2,j}}_{j = 1}^{\k}), &
\F(\seq{a_{3,j}}_{j = 1}^{\k}), & \ldots, & \F(\seq{a_{s,j}}_{j = 1}^{\k})
\end{bmatrix}'
\end{equation*}

We can then write the least squares optimisation problem for (\ref{eqn:reg})
in sample form as:

\begin{equation}\label{eqn:reg_solution}
\gh = \underset{\g}{\mathrm{argmin}} \| \y - \F_{\g}(\seq{\x_j}) \|^2_2
\end{equation}

Now assume WLOG that all the variables in this model are latent. The
problem specified by (\ref{eqn:reg}) then becomes:

\begin{equation}\label{eqn:reg_lv_a}
\ty = \F_{\g}(\seq{\tx_j}) + e
\end{equation}

According to Definition \(\ref{def:lv}\), we can approximate \(\ty\) and
all \(\tx_j\) using linear combinations of observed variables. Define
\(Y = \begin{bmatrix} y_1, & y_2, & y_3, & \ldots, & y_n \end{bmatrix}\)
and \(X_j = \begin{bmatrix} x_{1,j}, & x_{2,j}, & x_{3,j}, & \ldots, &
x_{m_j,j} \end{bmatrix}\) such that \(Y \w = \ty + r_y\) and \(X_j \o_j
= \tx_j + r_j\) for some unknown weight vectors \(\w\) and \(\o_j\) with
dimensions \(n \times 1\) and \(m_j \times 1\), respectively. Plugging
these expressions into (\ref{eqn:reg_lv_a}) and collecting the error terms
produces a regression problem of the form:

\begin{equation}\label{eqn:reg_lv_b}
\begin{split}
& Y \w = \F_{\g}(\seq{X_j \o_j}) + \epsilon \\
& \text{where } \epsilon = e + r_y + \F_{\g}(\seq{\tx_j}) - \F_{\g}(\seq{X_j \o_j})
\end{split}
\end{equation}

Here, \(\epsilon\) collects all the possible sources of regression error
in (\ref{eqn:reg_lv_a}): \(e\) is the residual noise from the functional
relationship \(\F_{\g}: \seq{\tx_j} \mapsto \ty\) under the true
regression parameters, \(r_y\) is the measurement noise from
approximating \(\ty\) based on \(Y\), and \(\F_{\g}(\seq{\tx_j}) -
\F_{\g}(\seq{X_j \o_j})\) summarises the reduction in the quality of
regression fit caused by the measurement noise in all \(\tx_j\).

When we compare the traditional regression problem (\ref{eqn:reg}) to the
latent variable regression problem (\ref{eqn:reg_lv_b}), two important
differences jump out. First, to the extent that the statistical
properties of any regression model are defined by the properties of its
error term, (\ref{eqn:reg_lv_b}) looks different from (\ref{eqn:reg}) because of all
the additional components in \(\epsilon\). By extension, the
optimisation techniques which are appropriate for (\ref{eqn:reg}) may not
always be applicable to (\ref{eqn:reg_lv_b}). However, a good counter-example
to this argument is provided by any traditional regression model in
which the sample measurements of the observed variables \(y\) and
\(x_j\) contain measurement uncertainty. In these cases one could just
as well write out \(y\) and \(x_j\) as \(y = \ty + r_y\) and \(x_j =
\tx_j + r_j\) for some ``true'' intrinsic processes \(\ty\) and
\(\tx_j\) with the measurement noise removed\footnote{An argument can be made that this counter-example represents
many research settings prevalent in finance and economics. For example,
asset prices are susceptible to noise trading, while macroeconomic
aggregates will generally have non-negligible measurement error.}. Plugging these
expressions back into (\ref{eqn:reg}) would produce the same error term as in
(\ref{eqn:reg_lv_b})\footnote{Under this assumption both (\ref{eqn:reg}) and (\ref{eqn:reg_lv_a}) can also
be interpreted as errors-in-variables (EIV) models
(\citep{Griliches-1970}). A thorough discussion of the statistical
properties of EIV models is beyond the scope of this paper, but the
reader is referred to \citep{Wolter-1982} as a starting point.}. This allows us to draw a strong methodological
parallel between traditional regression models and LVMs with limited
loss of generality. Formally speaking, for any intrinsic random
processes \(\ty\) and \(\tx_j\) governed by a hypothesised functional
relationship \(\F_{\g}\), any viable approach to estimating \(\F_{\g}\)
as a traditional regression model should also be applicable in a latent
variable setting as long as all the available sample measurements for
\(\ty\) and \(\tx_j\) are prone to measurement uncertainty.

The second difference between (\ref{eqn:reg}) and (\ref{eqn:reg_lv_b}) stems from
the extra parameters required for estimating (\ref{eqn:reg_lv_b}) (recall that
the coefficient vectors \(\w\) and \(\o_j\) are all initially unknown).
One could say that the asset pricing theory (APT) has largely succeeded
in estimating \(\o_j\) using a combination of economically motivated
heuristics\footnote{Arguably the most famous example of a heuristic approach to the
construction of asset pricing factors is \citep{Fama-French-1993} which
recombines a large cross-section of single stocks into 25 value-weighted
portfolios based on valuation and market capitalisation, and then uses
those portfolios as both the dependent variables and the building blocks
for the Fama-French three-factor model. A similar technique was
subsequently adopted in other prominent papers, such as
\citep{Asness-2013} which constructs 48 portfolios across asset classes
ranked on valuation and price momentum (the individual portfolios are
value-weighted in the case of single stocks and equal-weighted in other
cases).} and principal component analysis\footnote{For example, \citep{Bai-2006} uses PCA to show that the
Fama-French three factor model (FF3F) (\citep{Fama-French-1993}) is a
reasonable approximation of the latent factor space in the US equity
market; \citep{Lustig-2011} demonstrates that the first two principal
components in the currency market map to a theoretically motivated
two-factor model consisting of a ``dollar'' (market) factor and a
``slope'' (carry) factor; \citep{Diebold-2006} proposes a PCA-based
factor model for the term structure of interest rates; etc.}. On the other
hand, the APT methodology does not lend itself well to other research
settings. To illustrate, let us consider a generic asset pricing formula
for a representative test asset \(\ty_i\):

\begin{equation}\label{eqn:reg_lv_fm}
\ty_i = c_i + \sum_{j = 1}^{k}{\beta_{i,j} \tx_j} + e,
\quad i \in \mathbb{N}
\end{equation}

Here, \(\tx_j\) are the latent asset pricing factors, \(\beta_{i,j}\)
are asset \(i\)'s factor loadings, and \(c_i\) is the asset-specific
intercept term. This class of models is characterised by two unique
properties: a simple functional relationship, and a large cross-section
of investment instruments \(\ty_i\) which may or may not be latent in
the sense of being linear combinations\footnote{Historically, asset pricing tests have used single stocks,
bonds, currency pairs, regional benchmark indices, and even custom test
portfolios (\citep{Fama-French-1993}, \citep{Lustig-2011} and
\citep{Asness-2013} are some of the most famous examples).}. The key assumption is
that all \(\ty_i\) are driven by a small number of latent systematic
risk factors \(\tx_j\) and some idiosyncratic noise. A breakthrough in
estimating this class of models was achieved by \citep{Bai-2006} with a
proof that the space of the true unobserved factors \(\tx_j\) is
consistently spanned by the set of diffusion indices \(\ty^{di}_i\)
obtained from a PCA decomposition\footnote{It can be noted that several papers have subsequently applied
various extensions to the PCA methodology for the purpose of estimating
latent factors more effectively. For example, \citep{Boucher-2021} uses
a sparse (penalised) version of PCA to approximate latent factors for
the European equity market, \citep{Begusic-2020} uses a clustering
technique to identify globally and locally relevant latent factors, and
\citep{Lettau-2020} develops a procedure called Risk Premium PCA
(RP-PCA) for estimating latent factors that better capture
cross-sectional risk premia in an asset pricing context.} over all \(\ty_i\), as long as
the cross-section of test assets is large enough. This result allows the
researcher to reduce the validation process for any economically
motivated asset pricing model to a three-step procedure: calculate a set
of DI factors \(\{\ty^{di}_i\}\), propose a set of economically
motivated heuristic factors \(\{\tx^h_j\}\), and show that \(\tx^h_j\)
and \(\ty^{di}_i\) span the same space. Formal statistical tests for
this procedure have been proposed by \citep{Bai-2006},
\citep{Parker-2016} and \citep{Andreou-2025}, among others.

Unfortunately, this three-step approach becomes far less effective when
the number of test assets \(\ty_i\) is small and fixed\footnote{\citep{Ahn-2018} shows that CCA is preferable to PCA under these
circumstances.}, or when
the functional relationship no longer reduces to a simple linear model.
The latter is a big stumbling block in macroeconomics: Like APT,
macroeconomic theory has long entertained the notion of a limited number
of unobserved processes governing a larger set of observed macroeconomic
aggregates\footnote{One prominent example of this paradigm is the theory of business
cycles which goes back as far as \citep{Burns-Mitchell-1946}.}; however, the functional relationships in
macroeconomics can often be non-linear or have an unknown lead-lag
dimension. Empirically, \citep{Stock-2002} applies PCA to 215
macroeconomic variables in an attempt to approximate a set of latent
factors driving 8 key economic aggregates. The forecasting performance
of the resulting diffusion indices is compared to that of the
economically motivated observed predictors using an autoregressive model
with exogenous inputs (ARX) up to lag order \(q\). The underlying
regression model can be expressed as:

\begin{equation}\label{eqn:reg_arx_di}
y_{t+1} = c + \sum_{\tau = 0}^{q}{\phi_{\tau} y_{t-\tau}} +
  \sum_{j=1}^{\k} \sum_{\tau=0}^{q} {\beta_{j,\tau} \tx^{di}_{j,t-\tau}} + e
\end{equation}

The results obtained by \citep{Stock-2002} suggest that PCA-based
diffusion indices constructed from a large cross-section of
macroeconomic variables can produce more accurate forecasts for key
macroeconomic aggregates than the individual observed predictors.
However, unlike in APT, there is no longer any guarantee that these
diffusion indices produce an optimal factor model. This point is best
illustrated with same example that \citep{Stock-2002} uses in support of
PCA, namely, that of the empirical Phillips curve. Broadly speaking, the
Phillips curve postulates that inflation can be predicted by real
economic activity; however, there is no consensus about which measure of
economic activity is the most appropriate\footnote{Alternative measures of real economic activity include GDP,
employment, capacity utilisation, business surveys, and various activity
``nowcasts'', to name a few.}. One could say that
applying PCA to a broad set of real activity measures negates the
question of variable selection altogether because it allows the
researcher to extract all the relevant sources of variation from the
real activity data. However, PCA also reveals absolutely nothing about
the relevance of each of those sources of variation for inflation
specifically. Just like with the observed measures of real activity,
inflation can be influenced by any one of the diffusion indices of real
activity or, indeed, any linear combination thereof. What's more, if no
diffusion indices can be removed from the pool of candidate regressors
out of hand, PCA becomes somewhat redundant altogether (\ref{pca_redundancy}
expresses this argument more formally).

An effective estimation procedure for the broad class of latent variable
models defined by (\ref{eqn:reg_lv_b}) would need to address the shortcomings
of both PCA and domain-specific heuristic approaches. To this end, this
paper proposes a relatively simple solution: simultaneously estimate all
the coefficients in (\ref{eqn:reg_lv_b}) using the method of least squares.
Formally, define:

\begin{itemize}
\item \(n \geq 1\) proxy measurements for the dependent variable \(\ty\),
collected into a \(1 \times n\) row vector \(Y\)
\item \(m_j \geq 1\) proxy measurements for each explanatory variable
\(\tx_j\) where \(1 \leq j \leq \k\), collected into a sequence of
row vectors \(\seq{X_j} \equiv \seq{X_j}_{j=1}^{\k}\)
\item An \(s \times n\) matrix of sample observations \(\Y\) for the
variables in \(Y\) and a sequence of sample observation matrices
\(\seq{\X_j}\), each with dimensions \(s \times m_j\), for the
explanatory variable vectors \(X_j\).
\end{itemize}

The least-squares solution to (\ref{eqn:reg_lv_b}) in sample form is then
obtained by the vectors \(\gh\), \(\wh\) and \(\seq{\oh_j} \equiv
\seq{\oh_j}_{j=1}^{\k}\) which satisfy:

\begin{equation}\label{eqn:sdf_solution}
\gh,\wh,\seq{\oh_j} = \underset{\g,\w,\seq{\o_j}}{\mathrm{argmin}}
                      \| \Y \w - \F_{\g}(\seq{\X_j \o_j}) \|^2_2
\end{equation}

In practice, we would often need to impose constraints on \(\w\), \(\o\)
and/or \(\g\) because the economic interpretation of the model or the
underlying variables implies certain restrictions (e.g., an asset
pricing factor may have to be a zero-cost investment portfolio), or
because we wish to avoid trivial solutions such as \(\w = \bm{0}\). For
example, PCA and partial least squares (PLS) models
(\citep{Wold-1975,Wold-1982}) impose a unit length constraint on \(\w\)
and \(\o\), while CCA is implemented with unit variance constraints on
\(\Y\w\) and \(\X\o\). For now, let us define an arbitrary set of
constraint functions \(\{ g_i \colon \mathbb{R}^{M+n+p} \longrightarrow
\mathbb{R}, \, i \in \mathbb{N} \}\), where \(M = \sum_{j =
1}^{\k}{m_j}\) and \(p\) is the length of \(\g\), and write the full
constrained optimisation problem as:

\begin{equation}\label{eqn:sdf_solution_full}
\begin{split}
\underset{\g,\w,\seq{\o_j}}{\mathrm{min}} &
                      \| \Y \w - \F_{\g}(\seq{\X_j \o_j}) \|^2_2 \\
\text{s.t.} \quad &
  \{g_i \left( \g,\w,\seq{\o_j} \right) \geq 0: i \in \mathbb{N}\}
\end{split}
\end{equation}

This paper will henceforth refer to (\ref{eqn:sdf_solution_full}) as the
``Supervised Diffusion Framework'' (SDF) for latent variable modelling.

From a mathematical perspective, SDF can be viewed as a generalisation
over regression analysis on the one hand, and traditional variance
optimisation models on the other (Section \ref{sdf_const} further elaborates
on this point). From a practical perspective, the usefulness of SDF has
at least three facets. First, it address the problem of interpretability
inherent in PCA. Unlike PCA-based diffusion indices, supervised
diffusion indices (SDI) conform to a predefined functional relationship
in which the input variables would usually have an a priori economic
meaning\footnote{\citep{Ahn-2018} makes a similar argument with respect to latent
variables estimated by CCA. Section \ref{cca} shows that CCA is a special
type of a supervised diffusion model.}. Second, a supervised diffusion model (SDM) can be used
to improve the accuracy of empirical measurement for variables which are
not altogether latent yet not very accurately observed. For example, an
economic model describing the relationship between GDP and unemployment
could be used to better approximate both of these aggregates from sets
of alternative proxy measurements. Third, SDMs can produce synthetic
variables whose sole purpose is to track other random processes in
non-trivial ways. For example, an SDM can be used to build an investment
portfolio which tracks the market's expectation for a macroeconomic
index over multiple forecast horizons, or to devise a systematic trading
strategy which can be accurately predicted by a leading signal. The
concluding paragraph of Section \ref{sdf_const} further expands on this use
case by looking at SDMs through the lens of portfolio optimisation.
\subsection{SDF as a Variance Optimisation Framework}
\label{sdf_const}
The SDF formula is inherently a generalised variance optimisation
problem which subsumes a number of popular statistical techniques in
economics and finance as special cases. To demonstrate this, let us
explicitly add an intercept term \(c\) to equation (\ref{eqn:reg_lv_a}) and
rewrite it as:

\begin{equation}
\ty = c + \F_{\g} \left( \seq{\tx_j} \right) + e
\end{equation}

The solution for \(c\) in a least squares setting is well documented in
prior literature, but it is replicated below for completeness. Defining
a shorthand \(\yh := \F_{\g} \left( \seq{\tx_j} \right)\), as well as
the sample counterparts for \(\ty\) and \(\yh\) as \(\tmy\) and
\(\myh\), respectively, we can rewrite the optimisation problem
(\ref{eqn:sdf_solution}) as:

\begin{equation}\label{eqn:sdf_solution_short}
\min \| \tmy - \left( \1_s c + \myh \right)\|^2_2
    = \min \left( \tmy'\tmy + \1_s'\1_s c + \myh' \myh - 2 c \1_s' \tmy
    - 2 \tmy'\myh + 2  c \1_s' \myh \right)
\end{equation}

Here, \(s\) denotes the sample size of \(\tmy\) and \(\myh\), and
\(\1_s\) is a column vector of ones of sample length \(s\). Note that
this problem is convex with respect to \(c\), so an unconstrained
solution can be found by setting the partial derivative to zero:

\begin{equation}
\begin{split}
& \frac{\partial}{\partial \, c} \| \tmy - \left( \1_s c + \myh \right)\|^2_2
        = 2 \1_s c - 2 \tmy + 2 \myh \\
& \1_s c = \tmy - \myh
\end{split}
\end{equation}

Pre-multiplying both sides by \(\1_s' / s\) yields:

\begin{equation}\label{eqn:sdf_solution_c}
\ch = \frac{\1_s' \tmy}{s} - \frac{\1_s' \myh}{s} = \mean{\tmy} - \mean{\myh}
\end{equation}

where \(\mean{\tmy}\) and \(\mean{\myh}\) denote the sample means of
\(\tmy\) and \(\myh\), respectively. By plugging the unconstrained
solution for \(c\) back into (\ref{eqn:sdf_solution_short}) we arrive at a
simplified version of (\ref{eqn:sdf_solution_short}), namely:

\begin{equation}\label{eqn:sdf_solution_cov}
\min \| \left(\tmy - \1_s\mean{\tmy} \right) -  \left(\myh - \1_s\mean{\myh} \right)\|^2_2
    = \min \left( \Sm_{\ty} - 2 \Sm_{\ty\yh} + \Sm_{\yh} \right)
\end{equation}

Here, \(\Sm_A\) denotes the sample covariance matrix over random
variable vector \(A\), and \(\Sm_{AB}\) denotes the sample covariance
matrix from \(A\) to random variable vector \(B\) with the elements of
\(A\) as rows and the elements of \(B\) as columns\footnote{In theory, a constraint could be imposed on \(c\) which makes it
deviate from the solution given by (\ref{eqn:sdf_solution_c}). However, even
with a constraint on \(c\) the simplification provided by
(\ref{eqn:sdf_solution_cov}) remains meaningful with the caveat that something
other than the sample mean would need to be subtracted from \(\tmy\) or
\(\myh\) in equation (\ref{eqn:sdf_solution_cov}). In other words, if a
constraint is imposed on \(c\), \(\Sm\) will represent a biased sample
variance-covariance estimate for at least one of the variable vectors
involved. See \citep{Lettau-2020} for an example of when this may be
desirable.}. Note that in
this particular case, \(\Sm_{\ty}\), \(\Sm_{\ty\yh}\) and \(\Sm_{\yh}\)
all resolve to scalar values because both \(\ty\) and \(\yh\) are
univariate. Put simply, the SDF objective is nothing more than a
variance minimisation problem which subsumes a few popular models like
PCA and mean-variance portfolio optimisation as special cases. To
illustrate, consider a trivial SDM in which \(\F_{\g} \brac{
\seq{\tx_j}} = 0\). In this case, (\ref{eqn:sdf_solution_cov}) reduces to:

\begin{equation}\label{eqn:sdm_const}
\begin{split}
& \min \Sm_{\ty} \equiv \underset{\w}{\min} \, \w' \Sm_{Y} \w \\
\text{s.t. } &
  \{g_i \left( \w \right) \geq 0: i \in \mathbb{N}\}
\end{split}
\end{equation}

where \(\{g_i \left( \w \right) \geq 0: i \in \mathbb{N}\}\) is an
arbitrary set of constraints on \(\w\) (without the inferential
component, \(\o\) and \(\g\) become redundant). Here, the PCA objective
function\footnote{Technically speaking, the goal of PCA is to maximise variance
rather than to minimise it. However, objective (\ref{eqn:sdm_const}) is still
solved by PCA decomposition with the caveat that variance is minimised
by the last principal component and maximised by the first.} is obtained by applying a single constraint of the form
\(\w' \w = 1\). Similarly, if \(\Y\) is a matrix of historical asset
returns, a minimum variance portfolio is obtained with the constraint \(\1_{n}' \w = 1\), and a portfolio on the efficient frontier as in
\citep{Markowitz-1952} is produced with the help of a second constraint
of the form \(\mean{\Y} \w = r\) for a target return \(r\).

Conceptually, we can look at this result in three ways. First, SDF can
be described as a blend between traditional regression analysis and
variance optimisation models. SDF reduces a traditional regression
problem when all the input variables have one observed proxy each (i.e.,
when \(\w\) and all \(\o_j\) are scalars or multiples of standard basis
vectors), and it reduces to a traditional variance minimisation problem
if there is no inferential component \(\F_{\g} \left( \seq{\tx_j}
\right)\). Second, if we describe SDF as a ``supervised'' diffusion
methodology, we can conversely describe traditional variance
optimisation problems like PCA and mean-variance portfolio optimisation
as unsupervised diffusion models (once again, the ``supervised'' part
comes from the inferential component \(\F_{\g} \left( \seq{\tx_j}
\right)\) which affects the estimates of the latent variables by
endowing them with inferential properties).

Third, if traditional mean-variance portfolio optimisation can be
described as a trivial SDM, then conversely, non-trivial SDMs have a
place in investment management. Notably, a latent variable constructed
from investible assets under portfolio-style constraints becomes a
viable investment strategy in its own right. This opens up a wide range
of potential applications for SDMs in the investment industry, with
examples including asset allocation, systematic trading, and the
construction of tracker funds. As a starting point, the empirical study
in Section \ref{empirical_results} produces a ``market SDI''which can be
interpreted as a tracker portfolio for growth expectations in the United
States.
\section{Latent Variable Autoregression with Exogenous Inputs}
\label{larx}
The SDF formula can be used to turn any traditional regression model
into a supervised diffusion model (SDM) by allowing the input variables
to be latent. This section provides a tangible example of this process
by deriving an SDM counterpart of the ubiquitous autoregressive model
with exogenous inputs (ARX).

The ARX model is one of the most popular regression models in economics
and finance and in time series analysis more broadly. In ARX models, the
dependent is expressed as a linear function of its own past values and
the (present and) past values of one or more explanatory variables.
Special cases of the ARX model include autoregressive models with no
exogenous inputs (AR), lead-lag regression models with no autoregressive
element, and multiple linear regression models with no lag structure.

To derive a latent variable ARX model (LARX), we can start from the
standard ARX specification and assume that the dependent and all the
explanatory variables are latent. First, let us define the LARX problem
in functional form. Let \(v\) be a ``version iterator'' representing
different versions\footnote{Note that in LARX models, versions need not be synonymous with
time series lags. For example, \(\ty \equiv Y \w\) may represent the
return on an investment index and \(\tx_i \equiv X_i \w\) may be some
other property of the same index, such as market capitalisation as in
the ``size'' factor of \cite{Fama-French-1993}.} of latent variable \(j\), i.e., a variable
identified by LV weight vector \(\o_j\), and let the total number of
versions for variable \(j\) be \(V_j \geq 1\). Let \(K\) denote the
total number of unique exogenous variables \(\tx_j\) excluding versions,
i.e., the total number of unique weight vectors \(\o_j\). We can refer
to the ``autoregressive'' versions of the dependent variable as
\(\ta_v\), with the number of autoregressive versions denoted by the
capital \(V_a \geq 0\). A generic formula for the LARX model can then be
written as:

\begin{equation}\label{eqn:larx_a}
\ty = c + \sum_{v = 1}^{V_a} \phi_{v} \ta_{v} +
            \sum_{v = 1}^{V_1} \beta_{1,v} \tx_{1,v} +
            \sum_{v = 1}^{V_2} \beta_{2,v} \tx_{2,v} +
            \ldots + \sum_{v = 1}^{V_K} \beta_{K,v} \tx_{K,v} + e
\end{equation}

As a reminder, any input variable in this model becomes non-latent if
the corresponding weight vector only has one non-zero element, and the
entire problem reduces to a traditional ARX model when all the input
variables are non-latent.

Models of this kind are more conveniently solved in matrix form. The
matrix representation for (\ref{eqn:larx_a}) can be derived with the help of
the Kronecker product, traditionally denoted by ``\(\otimes\)'', and the
block-wise Kronecker product denoted by ``\(\odot\)'' as in
\citep{Khatri-Rao-1968}. The autoregressive terms can be written as:

\begin{equation*}
\sum_{v = 1}^{V_a} \phi_{v} \ta_{v} = A \pw + \mathbf{r}_a \p
\end{equation*}

where \(A = \begin{bmatrix}A_1, & A_2, & \cdots & A_{V_a}
\end{bmatrix}\) is a horizontal concatenation of the autoregressive
proxy vectors for the dependent, \(\p = \begin{bmatrix}\phi_1, & \phi_2,
& \cdots & \phi_{V_a} \end{bmatrix}'\) is a column vector containing the
respective autoregressive coefficients, and \(\mathbf{r}_a
= \begin{bmatrix}r_{a,1}, & r_{a,1}, & \cdots & r_{a,V_a}
\end{bmatrix}\) is a vector of approximation errors where \(r_{a,v} =
\ta_v - A_v \w\) for \(1 \leq v \leq V_a\). Similarly, for each
individual explanatory variable \(\tx_j\) we have:

\begin{equation*}
\sum_{v = 1}^{V_j} \beta_{j,v} \tx_{j,v} =
  X_j \left( \b_j \otimes \o_j \right)  + \mathbf{r}_j \b_j
\end{equation*}

with \(X_j = \begin{bmatrix}X_{j,1}, & X_{j,2}, & \cdots & X_{j,V_j}
\end{bmatrix}\), \(\b_j = \begin{bmatrix}\beta_{j,1}, & \beta_{j,2}, &
\cdots & \beta_{j,V_j} \end{bmatrix}'\), and \(\mathbf{r}_j\) a (\(1
\times V_j\)) vector of approximation errors for the different versions
of \(\tx_j\). The explanatory terms can be written out as:

\begin{equation*}
\sum_{j = 1}^{K} \sum_{v = 1}^{V_j} \beta_{j,v} \tx_{j,v}
    = \sum_{j = 1}^{K} X_j \left( \b_j \otimes \o_j \right) +
      \sum_{j = 1}^{K} \mathbf{r}_j \b_j
    = X \bo + \sum_{j = 1}^{K} \mathbf{r}_j \b_j
\end{equation*}

Here, \(X = \begin{bmatrix} X_1 | X_2 | \cdots | X_K\end{bmatrix}\) is a
row vector with \(K\) blocks corresponding to the individual \(X_j\),
\(\b = \begin{bmatrix} \b_1' | \b_2' | \cdots | \b_K'\end{bmatrix}'\) is
a column vector with \(K\) blocks for the individual \(\b_j\), and
\(\o\) has \(K\) blocks containing the individual \(\o_j\) such that
\(\o = \begin{bmatrix} \o_1' | \o_2' | \cdots | \o_K'\end{bmatrix}'\).

The complete LARX formula is then concisely defined in matrix form as:

\begin{equation} \label{eqn:larx_b}
Y \w = c + A \pw + X \bo + \epsilon
\end{equation}

where \(Y\), \(A\) and \(X\) comprise the underlying observed variable
space and \(\epsilon\) collects all the error terms.

Given a data sample of length \(s\) represented by observation matrices
\(\Y\), \(\A\) and \(\X\), the least squares optimisation problem for
the LARX model can be defined as

\begin{equation}\label{eqn:clarx_optim_a}
\underset{c,\p,\b,\w,\o}{\min} \, \| \Y \w - \1_s c - \A \pw - \X \bo \|^2_2
\end{equation}

In terms of optimisation constraints, the only requirement for a minimal
LARX implementation is a scaling constraint on \(\w\) to avoid the
trivial solution given by \(\w = \0\). However, in view of the
investment management use case discussed in the previous section, we can
start by deriving a constrained implementation of the LARX model (call
it CLARX) with full Markowitz-style constraints on the variance and sum
of weights of each latent variable. The ``unconstrained'' LARX model can
then be derived by setting the appropriate Lagrange multipliers to zero.
The variance constraint on the dependent takes the form of \(\w' \Sm_{Y}
\w = \s2_y\). The variance constraint on explanatory variable \(j\)
takes the form of \(\o_j' \Sm_{X_{j,c_j}} \o_j = \s2_j\), where \(c_j\)
is an arbitrarily chosen version (lag) of \(\tx_j\). The sum-of-weights
constraints take the form of \(\1_n' \w = l_y\) and \(\1_{m_j}' \o_j =
l_j\), for some scalar values \(l_y\) and \(l_j\). The complete
constrained optimisation problem then becomes\footnote{To the author's best knowledge, the closest precedent for this
optimisation problem in prior literature is the EDACCA model defined in
\cite{Xu-2024}. The EDACCA model has a number of important differences,
including a different set of constraints and the use of a single
explanatory variable.}:

\begin{equation}\label{eqn:clarx_optim_b}
\begin{split}
\underset{c,\p,\b,\w,\o}{\min} \, & \| \Y \w - \1_s c - \A \pw - \X \bo \|^2_2 \\[10pt]
\text{s.t. } & \w' \Sm_{Y} \w = \s2_y, \quad
  \o_j' \Sm_{X_{j,c_j}} \o_j = \s2_j \text{ for } 1 \leq j \leq K, \\[10pt]
& \1_n' \w = l_y, \quad \1_{m_j}' \o_j = l_j \text{ for } 1 \leq j \leq K
\end{split}
\end{equation}

The solution for \(c\) is covered by section \ref{sdf_const}. In this case it
reduces to:

\begin{equation}\label{eqn:clarx_solution_const}
\ch = \overline{\Y} \w - \overline{\A} \pw - \overline{\X} \bo
\end{equation}

where \(\overline{\Y}\), \(\overline{\A}\) and \(\overline{\X}\) are row
vectors containing the column-wise sample means of \(Y\), \(A\) and
\(X\). Plugging this solution back into (\ref{eqn:clarx_optim_b}) and
expanding, we obtain a simplified optimisation problem of the form:

\begin{align} \label{eqn:clarx_optim_c}
\begin{split}
\underset{\p,\b,\w,\o}{\min} \, &  \w' \Sm_{Y} \w + \pw' \Sm_{A} \pw + \bo' \Sm_{X} \bo \\[5pt]
        & - 2 \left[ \w' \Sm_{YA} \pw + \w' \Sm_{YX} \bo - \pw' \Sm_{AX} \bo \right]
\end{split} \\[15pt]
\begin{split}
\text{s.t. } & \w' \Sm_{Y} \w = \s2_y, \quad
  \o_j' \Sm_{X_{j,c_j}} \o_j = \s2_j \text{ for } 1 \leq j \leq K, \\[10pt]
& \1_n' \w = l_y, \quad \1_{m_j}' \o_j = l_j \text{ for } 1 \leq j \leq K
\end{split} \nonumber
\end{align}

This is a convex problem with equality constraints, which means it can
be solved using the method of Lagrange multipliers (LM). As with the
rest of the model, we can represent the LM terms in matrix form with the
help of the following vectors:

\begin{align*}
\ts2_x & = \begin{bmatrix}
    \s2_{x,1}, &
    \s2_{x,2}, &
    \s2_{x,3}, &
    \ldots &
    \s2_{x,K}
\end{bmatrix}', \quad
& \tl_x = \begin{bmatrix}
    \l_{x,1}, &
    \l_{x,2}, &
    \l_{x,3}, &
    \ldots &
    \l_{x,K}
\end{bmatrix}' \\[10pt]
\lp & = \begin{bmatrix}
    l_{p,1}, &
    l_{p,2}, &
    l_{p,3}, &
    \ldots &
    l_{p,K}
\end{bmatrix}', \quad
& \tl_p = \begin{bmatrix}
    \l_{p,1}, &
    \l_{p,2}, &
    \l_{p,3}, &
    \ldots &
    \l_{p,K}
\end{bmatrix}'
\end{align*}

Here, \(\ts2_x\) is a \(K \times 1\) column vector of variance targets
for the chosen versions of \(\tx_j\); \(\lp\) is a \(K \times 1\) column
vector of sum-of-weights targets for the respective LV weight vectors
\(\o_j\); and \(\tl_x\) and \(\tl_p\) are the corresponding vectors of
LM coefficients. The full constraint terms for the Lagrangian function
can then be defined in matrix form with the help of the block-wise
direct sum operator introduced in Section \ref{blockwise_direct_sum}. The
sum-of-weights constraints become:

\begin{equation*}
\odt \1_{\o} = \lp
\end{equation*}

where \(\1_{\o}\) is a column vector of ones with the same block
structure as \(\o\).

For the variance constraints we require an indexer to identify which
version of each \(\tx_j\) has a variance constraint assigned to it. Let
\(\seq{\u_j}\) be a sequence of \(K\) logical vectors, each of size
\(V_j\) (i.e., the number of versions of \(\tx_j\)) with a value of 1 in
position of \(c_j\), i.e., the version of \(\tx_j\) that has a variance
constraint, and zeros elsewhere. The column vector \(\u := \left[
\seq{\u_j} \right]\) is a vertical concatenation of \(\seq{\u_j}\) which
has the same size and block structure as the coefficient vector \(\b\).
We also need a block-diagonal matrix \(\Sm_X^d\) containing the
covariance matrices of the individual \(X_{j,v_j}\) along the diagonal,
but no covariances across different variables or versions. Once again,
this can be done with the help of the block-wise direct sum operator:

\begin{equation*}
\begin{split}
& \Smd_X = \left[ \left( \X - \overline{\X} \right)^{\oplus} \right]'
\left( \X - \overline{\X} \right)^{\oplus}
= \begin{pmatrix}
\Smd_{X_1}, & \0, & \cdots & \0 \\
\0, & \Smd_{X_2}, & \cdots & \0 \\
\vdots & \vdots & \ddots & \vdots \\
\0, & \0, & \cdots & \Smd_{X_K}
\end{pmatrix} \\
\text{with} \quad &
\Smd_{X_j} = \left[ \left( \X_j - \overline{\X_j} \right)^{\oplus} \right]'
           \left( \X_j - \overline{\X_j} \right)^{\oplus}
       = \begin{pmatrix}
            \Sm_{X_{j,1}}, & \0, & \cdots & \0 \\
            \0, & \Sm_{X_{j,2}}, & \cdots & \0 \\
            \vdots & \vdots & \ddots & \vdots \\
            \0, & \0, & \cdots & \Sm_{X_{j,V_j}}
       \end{pmatrix} \quad \text{for } 1 \leq j \leq K
\end{split}
\end{equation*}

The variance constraints on all \(\tx_j\) can then be expressed using a
single quadratic form:

\begin{equation*}
\left[ \uo' \right]^{\oplus} \Smd_X \uo = \ts2_x
\end{equation*}

and the full Lagrangian function for (\ref{eqn:clarx_optim_c}) can be written
as:

\begin{equation} \label{eqn:clarx_lagr}
\begin{split} 
  \Lagr( & \p,\b,\w,\o,\l_y,\l_l,\tl_x,\tl_p) = \w' \Sm_{Y} \w +
        \pw' \Sm_{A} \pw + \bo' \Sm_{X} \bo \\[12pt]
       & - 2 \w' \Sm_{YA} \pw - 2 \w' \Sm_{YX} \bo + 2 \pw' \Sm_{AX} \bo
       + \l_y \left( \w' \Sm_{Y} \w - \s2_y \right) \\[12pt]
       & + \l_l \left( \1_n' \w - l_y \right)
       + \tl_x' \left\{ \left[ \uo' \right]^{\oplus} \Smd_X \uo - \ts2_x \right\}
       + \tl_p' \left[ \odt \1_{\o} - \lp \right]
\end{split}
\end{equation}

Note that in expanded form the terms under \(\tl_x\) and \(\tl_p\)
resolve to:

\begin{align*}
& \tl_x' \left\{ \left[ \uo' \right]^{\oplus} \Smd_X \uo - \ts2_x \right\} =
    \sum_{j=1}^{K} \l_{x,j} \left( \, \o_j' \Sm_{X_{j,c_j}} \o_j - \s2_{x,j} \right)\\
& \tl_p' \left[ \odt \1_{\o} - \lp \right] =
    \sum_{j = 1}^K \l_{p,j} \left(\o_j' \1_{m_j} - l_{p,j} \right)
\end{align*}

where \(m_j\) represents the number of observed variables used to
approximate \(\tx_j\).

Like both mean-variance portfolio optimisation and linear least-squares
regression, the Lagrangian function specified by (\ref{eqn:clarx_lagr}) is
convex and can be solved by setting various partial derivatives to zero.
The derivations for the individual coefficient vectors and Lagrange
multipliers are somewhat involved and deferred to \ref{derivation_w_o_p_g}
and \ref{derivation_ry_rl}. The solution is given by a fixed point over the
following system of equations:

\begin{subnumcases}{\hspace{-5em}\label{eqn:clarx_solution_full}}
\wh = \frac{1}{\ry} \left[ \Sm_{Y}^{-1} \left( \tv_1 + \tv_2 \right)
   - \r_l \Sm_{Y}^{-1} \1_n \right]
   \label{eqn:clarx_fp_1}  \\[7pt]
\oh = \left[ \bi' \Sm_{X} \bi + \M_2 \brac{\tl_x \odot I_{\o}} \right]^{-1}
     \left[ \tv_3 - \frac{1}{2} \t1od \tl_p \right] 
    \label{eqn:clarx_fp_2} \\[7pt]
\ph = \left[ \wi' \Sm_{A} \wi \right]^{-1}
     \wi' \left[ \Sm_{AY} \w - \Sm_{AX} \bo \right]
     \label{eqn:clarx_fp_3} \\[7pt]
\bh = \left[ \oi' \Sm_{X} \oi \right]^{-1}
      \oi' \left[ \Sm_{XY} \w - \Sm_{XA} \pw \right]
     \label{eqn:clarx_fp_4} \\[7pt]
\ry = \frac{ \left( n \w - l_y \1_n \right)' \left( \tv_1 + \tv_2 \right) }
             {n \s2_y - l_y \1_n' \Sm_{Y} \w}
     \label{eqn:clarx_fp_5} \\[7pt]
\r_l = \frac{1}{n} \left[ \1_n' \left( \tv_1 + \tv_2 \right) -
                          \ry \, \1_n' \Sm_{Y} \w \right]
     \label{eqn:clarx_fp_6} \\[7pt]
\tl_x = \left[ \M_1 \V - \P \left( \t1od \right)' \M_2 \od \right]^{-1}
  \left( \od \M_1 - \t1od \P \right)' \left( \tv_3 - \tv_4 \right)
     \label{eqn:clarx_fp_7} \\[7pt]
\tl_p = 2 \M_1^{-1} \left( \t1od \right)' 
        \left[ \left( \tv_3 - \tv_4 \right) - \M_2 \od \tl_x \right] 
     \label{eqn:clarx_fp_8}
\end{subnumcases}

with the shorthand notations:

\begin{flalign*}
& \1_{\o} \text{ a column vector of ones with the same length and block structure as } \o & \\
& \tv_1 = \left[
    \pi' \Sm_{AY} +
    \Sm_{YA} \pi -
    \pi' \Sm_{A} \pi
\right] \w & \\
& \tv_2 = \left[\Sm_{YX} - \pi' \Sm_{AX} \right] \bo &  \\
& \tv_3 = \bi' \left[ \Sm_{XY} - \Sm_{XA} \pi \right] \w & \\
& \tv_4 = \bi' \Sm_{X} \bio & \\
& \V = \diag(\ts2_x) & \\
& \P = \diag(\lp) & \\
& \M_1 = \left( \t1od \right)' \t1od & \\
& \M_2 = \ui' \Smd_X \ui & \\
& \ry = 1 + \ly, \quad \r_l = \frac{\l_l}{2} &
\end{flalign*}

In the above expressions, \(I_a\) represents an identity matrix of
scalar size \(a\), while \(I_{\a}\) represents an identity matrix with
the same size and block structure as that of the vector \(\a\) (bold).

Here, it must be noted that the derivation steps required to arrive at
(\ref{eqn:clarx_solution_full}) follow directly from the properties of the
block-wise direct sum operator and the block-wise Kronecker product for
vectors introduced in Section \ref{blockwise_direct_sum}, \ref{bds_vs_mmul} and
\ref{block_kron_factorisation}. A thorough discussion of the relevance of
these results for the broader field of linear algebra is beyond the
scope of this paper; however, the derivations found in
\ref{derivation_w_o_p_g} and \ref{derivation_ry_rl} generalise rather easily to
other use cases in matrix calculus, including a larger class of
Lagrangian optimisation problems in which an arbitrary sequence of
coefficient vectors needs to be estimated in the presence of
interactions and case-by-case constraints.

In practice, the solution to (\ref{eqn:clarx_solution_full}) can be obtained
by fixed point iteration, which has sub-linear computational complexity
if the underlying regression problem is well-specified. Initial guesses
are required for \(\w\), \(\o\) and \(\p\). Equations can be estimated
in the same order as shown above. The first iteration can start at
(\ref{eqn:clarx_fp_4}). Four matrix inversions need to be recalculated
at each iteration step (namely,
(\ref{eqn:clarx_fp_2})-(\ref{eqn:clarx_fp_4}) and
(\ref{eqn:clarx_fp_7})) because the corresponding matrices are derived
from the estimated coefficient vectors; however, equations
(\ref{eqn:clarx_fp_3}) and (\ref{eqn:clarx_fp_4}) can be reformulated in
terms of the Moore-Penrose inverses (\citep{Penrose-1955,Bjerhammar-1951,Moore-1920}) of the matrices \((\A - \overline{\A})
\wi\) and \((\X - \overline{\X}) \oi\), respectively, which makes them
solvable by SVD.
\subsection{The Unconstrained Case and the Intuitive Interpretation}
\label{larx_ols_interpretation}
The constrained solution to the LARX model given by
(\ref{eqn:clarx_solution_full}) may look somewhat complex at first. However,
it simply says that the objective function is minimised when each of the
four coefficient vectors (\(\w\), \(\o\), \(\p\), \(\b\)) comes as close
to a least-squares solution to (\ref{eqn:larx_b}) as possible without
violating a portfolio constraint. We can see this more clearly by
removing the portfolio-style constraints from the model and only keeping
the one constraint required for model convergence\footnote{A scaling constraint has to be imposed on the dependent weight
vector \(\w\) to prevent the problem from converging to the trivial
solution given by \(\w = \0\). Several possibilities exist for such a
constraint, the most popular ones being a length constraint \(\w' \w =
k\) and a unit variance constraint \(\w' \Sm_{Y} \w = \s2\), where \(k\)
and \(\s2\) can be arbitrary scalars but are usually set to unity. A
unit length constraint is used in PCA and PLS; a unit variance
constraint is used in CCA. This paper chooses a unit variance constraint
for several reasons, including the ability to draw a parallel between
LARX and CCA, the usefulness of a variance constraint in portfolio
management settings, and the better interpretability of the concept of
variance in finance and economics more generally.}, namely, that
on the variance of \(\ty\). Setting the corresponding Lagrange
multipliers (\(\l_l\), \(\tl_x\) and \(\tl_p\)) to zero, we arrive at
the following (unconstrained) solution:

\begin{subnumcases}{\hspace{-5em}\label{eqn:larx_solution_full}}
\wh = \frac{1}{\ry} \Sm_{Y}^{-1} \left( \tv_1 + \tv_2 \right)
   \label{eqn:larx_fp_1}  \\[10pt]
\oh = \left[ \bi' \Sm_{X} \bi \right]^{-1} \bi' \left[ \Sm_{XY} - \Sm_{XA} \pi \right] \w
    \label{eqn:larx_fp_2} \\[10pt]
\ph = \left[ \wi' \Sm_{A} \wi \right]^{-1}
     \wi' \left[ \Sm_{AY} \w - \Sm_{AX} \bo \right]
     \label{eqn:larx_fp_3} \\[10pt]
\bh = \left[ \oi' \Sm_{X} \oi \right]^{-1}
     \oi' \left[ \Sm_{XY} - \Sm_{XA} \pi \right] \w
     \label{eqn:larx_fp_4} \\[10pt]
\ry = \frac{ \w' \left( \tv_1 + \tv_2 \right) }{\s2_y}
     \label{eqn:larx_fp_5}
\end{subnumcases}

with \(\tv_1\), \(\tv_2\) and \(\ry\) all defined as before. In this
model, each of the four coefficient vectors \(\wh\), \(\oh\), \(\ph\),
\(\bh\) can be expressed as an OLS solution to a linear regression
problem of the form \(a = c + B \g_i + e\), where both \(a\) and \(B\)
are derived using the other three vectors. For example, the coefficient
\(\bh\) solves:

\begin{equation*}
\begin{split}
& \tr_a = c + \tX \b + e \\
\text{with } & \tX := \begin{bmatrix} \tx_{1,1}, & \tx_{1,2}, & \ldots & \tx_{1, V_1}, & \ldots & \tx_{K, V_k} \end{bmatrix}
\quad \text{and} \quad \tr_a := \ty - \sum_{v = 1}^{V_a} \phi_v \ta_v
\end{split}
\end{equation*}

Here, \(\tX\) is a vector of all explanatory variables in latent form,
which can be estimated by \(X \brac{I_{\b} \odot \oh}\) if \(\oh\) is
known. Similarly, the variable \(\tr_a\) is the autoregressive residual
of \(\ty\) which can be estimated by \(Y \wh - A \brac{\ph \otimes \wh}\). In other words, \(\bh\) solves a regression problem of the form \(a
= c + B \g_i + e\) where \(a := Y \wh - A \brac{\ph \otimes \wh}\) and
\(B := X \brac{I_{\b} \odot \oh}\).

Similar transformations can be applied to derive the other coefficient
vectors. The corresponding formulas for \(a\) and \(B\) are provided in
Table \ref{table:larx_ols_matrices} below:

\begin{table}[ht]
\centering
\caption{OLS Inputs for LARX Coefficients}
\label{table:larx_ols_matrices}
\begin{threeparttable}
  \begin{tabular}{CCC}
  \toprule
  \text{Coefficient vector } \g_i & \text{Dependent variable } a & \text{Explanatory vector } B \\
  \midrule
  \w & X \brac{\bh \odot \oh} & Y - A \brac{\ph \otimes I_n} \\
  \p & Y \wh - X \brac{\bh \odot \oh} & A \brac{ I_{V_a} \otimes \wh } \\
  \o & Y \wh - A \brac{\ph \otimes \wh} & X \brac{ \bh \odot I_{\o} } \\
  \b & Y \wh - A \brac{\ph \otimes \wh} & X \brac{ I_{\b} \odot \oh } \\
  \bottomrule
  \end{tabular}
  
\begin{tablenotes}[flushleft]

 \item \footnotesize Note: Formulas for a dependent variable \(a\) and
 an explanatory vector \(B\) which conditionally produce a standard
 linear regression model \(a = c + B \g_i + e\) for which the given
 coefficient vector \(\gh_i\) provides the OLS solution.

\end{tablenotes}

\end{threeparttable}
\end{table}

This observation is useful in two respects. First, it provides an
alternative interpretation to the LARX model and its regression
parameters. Second, it creates a convenient shortcut for the purposes of
heuristic feature selection. If a coefficient vector can be estimated by
OLS, then so can its standard error, with the caveat that both are
obtained conditionally. Conditional standard errors fall short of
providing a full picture of model uncertainty, but they can shed some
light on the relative statistical significance of the individual
elements in each coefficient vector. This should allow the researcher to
perform heuristic feature selection for the LARX model in the same
manner as one would do with a standard least-squares regression.
\section{Special Cases of the LARX Model}
\label{larx_special_cases}
Like ARX, the LARX methodology subsumes several regression models as
special cases, and some of these special cases have interesting
mathematical properties in their own right. The main special cases of
the LARX model are presented in Table \ref{table:larx_summary}. Three of them
are examined more closely in this section.

\afterpage{
\begin{landscape}
\begin{table}[p]
\caption{Special Cases of the LARX model}
\label{table:larx_summary}
\renewcommand{\arraystretch}{2}
\begin{adjustbox}{width=\textwidth}
\begin{threeparttable}
  \begin{tabular}{p{0.2\textwidth} p{0.55\textwidth} p{0.5\textwidth} p{0.45\textwidth}}

  &
  \thead{LARX: Latent Variable Autoregression with Exogenous Inputs} &
  \thead{3LR: Latent Variable Lead-Lag Regression} &
  \thead{LAR: Latent Variable Autoregression} \\

  \toprule

  Description &

    The most generic implementation -- see Section \ref{larx}. &

    No autoregressive term. An SDM counterpart to a multivariate
    lead-lag regression. &

    No exogenous inputs. An SDM equivalent to the traditional
    autoregressive (AR) model. \\

  \midrule

  Regression formula &
  \( \ty = c + \sum_{v = 1}^{V_a} \phi_{v} \ta_{v} +
                \sum_{j = 1}^{K} \sum_{v = 1}^{V_j} \beta_{j,v} \tx_{j,v} + e \) & 
  \( \ty = c + \sum_{j = 1}^{K} \sum_{v = 1}^{V_j} \beta_{j,v} \tx_{j,v} + e \) & 
  \( \ty = c + \sum_{v = 1}^{V_a} \phi_{v} \ta_{v} + e \) \\

  Formula in matrix form & 
      \(Y \w = c + A \pw + X \bo + \epsilon \) &
      \(Y \w = c + X \bo + \epsilon \) &
      \(Y \w = c + A \pw + \epsilon \) \\

  \midrule

  Solution (optional constraints in \color{red}{red}) &

  \( \left\{ \begin{array}{l}
    \wh = \frac{1}{\ry} \Sm_{Y}^{-1} \left( \tv_1 + \tv_2 \right)
      \color{red}{- \frac{\r_l }{\ry} \Sm_{Y}^{-1} \1_n} \\
    \oh = \left[ \bi' \Sm_{X} \bi \color{red}{+ \M_2 \brac{\tl_x \odot I_{\o}}} \color{black}{} \right]^{-1}
        \left[ \tv_3 \color{red}{- \frac{1}{2} \t1od \tl_p} \color{black}{} \right]  \\
    \ph = \left[ \wi' \Sm_{A} \wi \right]^{-1}
        \wi' \left[ \Sm_{AY} \w - \Sm_{AX} \bo \right] \\
    \bh = \left[ \oi' \Sm_{X} \oi \right]^{-1}
          \oi' \left[ \Sm_{XY} \w - \Sm_{XA} \pw \right] \\
    \ry = \frac{ \left( n \w \color{red}{ - l_y \1_n } \color{black}{} \right)' \left( \tv_1 + \tv_2 \right) }
                {n \s2_y \color{red}{- l_y \1_n' \Sm_{Y} \w}} \\
    \color{red}{\r_l = \frac{1}{n} \left[ \1_n' \left( \tv_1 + \tv_2 \right) -
                              \ry \, \1_n' \Sm_{Y} \w \right] }\\
    \color{red}{\tl_x = \left[ \M_1 \V - \P \left( \t1od \right)' \M_2 \od \right]^{-1}
      \left( \od \M_1 - \t1od \P \right)' \left( \tv_3 - \tv_4 \right) } \\
    \color{red}{\tl_p = 2 \M_1^{-1} \left( \t1od \right)' 
            \left[ \left( \tv_3 - \tv_4 \right) - \M_2 \od \tl_x \right] }
    \end{array} \right.
  \) &

  \( \left\{ \begin{array}{l}
    \wh = \frac{1}{\ry} \Sm_{Y}^{-1} \Sm_{YX} \bo
      \color{red}{- \frac{\r_l }{\ry} \Sm_{Y}^{-1} \1_n} \\
    \oh = \left[ \bi' \Sm_{X} \bi \color{red}{+ \M_2 \brac{\tl_x \odot I_{\o}}} \color{black}{} \right]^{-1}
        \left[ \tv_3 \color{red}{- \frac{1}{2} \t1od \tl_p} \color{black}{} \right]  \\
    \bh = \left[ \oi' \Sm_{X} \oi \right]^{-1} \oi' \Sm_{XY} \w \\
    \ry = \frac{ \left( n \w \color{red}{ - l_y \1_n } \color{black}{} \right)' \Sm_{YX} \bo }
                {n \s2_y \color{red}{- l_y \1_n' \Sm_{Y} \w}} \\
    \color{red}{\r_l = \frac{1}{n} \left[ \1_n' \left( \tv_1 + \tv_2 \right) -
                              \ry \, \1_n' \Sm_{Y} \w \right] }\\
    \color{red}{\tl_x = \left[ \M_1 \V - \P \left( \t1od \right)' \M_2 \od \right]^{-1}
      \left( \od \M_1 - \t1od \P \right)' \left( \tv_3 - \tv_4 \right) } \\
    \color{red}{\tl_p = 2 \M_1^{-1} \left( \t1od \right)' 
            \left[ \left( \tv_3 - \tv_4 \right) - \M_2 \od \tl_x \right] }
  \end{array} \right.
  \) &

  \( \left\{ \begin{array}{l}
    \wh = \frac{1}{\ry} \Sm_Y^{-1} \tv_1 \color{red}{- \frac{\r_l}{\ry} \1_n } \\
    \ph = \left[ \wi' \Sm_{A} \wi \right]^{-1} \wi' \Sm_{AY} \w \\
    \ry = \frac{ n \w' \tv_1 \color{red}{- l_y \1_n' \tv_1} }
                {n \s2_y \color{red}{- l_y \1_n' \Sm_{Y} \w} } \\
    \color{red}{\r_l = \frac{1}{n} \left[\1_n' \tv_1 - \ry \, \1_n' \Sm_{Y} \w \right]}
  \end{array} \right.
  \) \\

  \midrule

  Shorthand notations &

  \( \begin{array}{l}
    \1_{\o} \text{ a column vector of ones with the same length and block structure as } \o \\
    \tv_1 = \left[
        \pi' \Sm_{AY} +
        \Sm_{YA} \pi -
        \pi' \Sm_{A} \pi
    \right] \w \\
    \tv_2 = \left[\Sm_{YX} - \pi' \Sm_{AX} \right] \bo \\
    \tv_3 = \bi' \left[ \Sm_{XY} - \Sm_{XA} \pi \right] \w \\
    \tv_4 = \bi' \Sm_{X} \bio \\
    \V = \diag(\ts2_x), \, \P = \diag(\lp), \,
    \M_1 = \left( \t1od \right)' \t1od, \, \M_2 = \ui' \Smd_X \ui \\
    \end{array}
  \) &

  \( \begin{array}{l}
    \tv_3 = \bi' \Sm_{XY} \w \\
    \tv_4 = \bi' \Sm_{X} \bio \\
    \V = \diag(\ts2_x), \quad \P = \diag(\lp) \\
    \M_1 = \left( \t1od \right)' \t1od, \quad \M_2 = \ui' \Smd_X \ui \\
    \ry = 1 + \ly, \quad \r_l = \frac{\l_l}{2}
  \end{array}
  \) & 

  \( \tv_1 \) as in the LARX model \\

  \bottomrule

  &
  \rule{0pt}{40pt}\thead{LSR: Latent Shock Regression} &
  \thead{LVMR/CCA: Latent Variable Multiple Regression} &
  \thead{LAR(1)/CAA: First-order LAR Model} \\

  \toprule

  Description &

  A lead-lag regression between an observed dependent and a single
  latent explanatory. A parsimonious alternative to a traditional
  multivariate lead-lag regression model. See Section \ref{lsr}. &

  Linear regression between a latent dependent and one or more (latent
  or non-latent) explanatory variables. Each latent variable enters with
  one version (no duplicate weight vectors). See Section \ref{cca}. &

  A latent variable autoregressive model of order LAR(1). Solvable by
  Canonical Autocorrelation Analysis (CAA) when \(\Sm_Y = \Sm_A \)
  (covariance stationary proxy vector). See Section \ref{caa}.
  \\

  \midrule

  Regression formula &

  \( y_t = c + \sum_{\tau = 1}^{F} \beta_{\tau} L^{\tau} \tx_t + e_t \)
  \hspace{1.5em} where \(L\) denotes the lag operator & 

  \( \ty = c + \sum_{j = 1}^{K} \beta_{j} \tx_{j} + e \) & 

  \( \ty_t = c + \phi \ty_{t-1} + e_t \) \\

  Formula in matrix form & 

      \( y_t = c + X \bko + \epsilon_t \) \hspace{3em}
      with \( X =
      \begin{bmatrix} X_{t-1}, & X_{t-2}, & \ldots & X_{t-F} \end{bmatrix}
      \)  &

      \( Y \w = c + X \o + \epsilon \) \hspace{0.5em} where \(X\), \(\o\)
      are vector concatenations of \(\seq{X_j}\), \(\seq{\o_j}\) &

      \( Y \w = c + \phi A \w + \epsilon  \) \\

  \midrule

  Solution &

  \( \left\{ \begin{array}{l}
      \oh = \left[ \bki' \Sm_{X} \bki \right]^{-1} \bki' \Sm_{Xy} \\
      \bh = \left[ \oki' \Sm_{X} \oki \right]^{-1} \oki' \Sm_{Xy}
    \end{array} \right.
  \) &

  \makecell[l]{
   \vspace{0.2em}
   \( \left\{ \begin{array}{l}
      \wh = \frac{1}{\ry} \Sm_{Y}^{-1} \Sm_{YX} \o \\
      \oh = \Sm_{X}^{-1} \Sm_{XY} \w \\
      \ry = \frac{ \w' \, \Sm_{YX} \o }{\s2_y}
  \end{array} \right. \vspace{0.2em}
  \) \\or: \hspace{1em} \(
    \ry \wh = \Sm_{Y}^{-1} \Sm_{YX} \Sm_{X}^{-1} \Sm_{XY} \w
  \)
  \vspace{0.2em}
  } &

  \makecell[l]{
   \vspace{0.2em}
   \( \left\{ \begin{array}{l}
      \wh = \frac{1}{\ry} \Sm_Y^{-1} \phi
            \left[\Sm_{AY} + \Sm_{YA} - \phi \Sm_{A} \right] \w \\
      \phih = \frac{\w' \Sm_{AY} \w}{\w' \Sm_{A} \w} \\
      \ry = \frac{\phi }{\s2_y} \w' \left[\Sm_{AY} + \Sm_{YA} - \phi \Sm_{A}
\right] \w
  \end{array} \right. \vspace{0.2em}
  \) \\or: \hspace{1em} \(
    \phi \wh =  \frac{1}{2} \Sm_Y^{-1} \brac{ \Sm_{AY} + \Sm_{YA}} \w
  \) \hspace{1em} if \hspace{1em} \(\Sm_Y = \Sm_A \)
  \vspace{0.2em}
  }\\

  \bottomrule

  \end{tabular}

\end{threeparttable}
\end{adjustbox}
\end{table}
\end{landscape}
}
\subsection{LSR: A Parsimonious Alternative to a Traditional Lead-Lag Regression}
\label{lsr}
Just like the traditional ARX model, the LARX model is estimated by the
method of least squares applied to a collection of observed variables.
The key difference lies in how the regression coefficients are mapped to
the observed data. ARX models always assign a unique response
coefficient to each variable, whereas LARX models allow some
coefficients to enter the equation more than once as a side-effect of
the interaction between coefficient vectors and latent variable weight
vectors. As a result, LARX models will often be more parsimonious than
traditional regression models applied to the same dataset.

This difference is best exemplified by a class of models in which an
observed dependent variable \(y\) enters as a function of \(V\) versions
of a single latent explanatory variable \(\tx\) estimated using a \(\brac{1 \times m}\) vector of observed proxy variables, \(X\). In a
time series context, the version iterator \(v\) is more intuitively
described as a lag iterator \(\tau\) relative to a time subscript \(t\),
so the corresponding regression problem becomes:

\begin{equation}\label{eqn:lsr_a}
\begin{split}
& y_t = c + \sum_{\tau = 1}^{F} \beta_{\tau} \tx_{t-\tau} + e_t
    \quad \text{ or, equivalently:} \\
& y_t = c + \sum_{\tau = 1}^{F} \beta_{\tau} L^{\tau} \tx_{t} + e_t
    \quad \text{where } L \text{ is the lag operator}
\end{split}
\end{equation}

As before, the latent variable \(\tx\) is approximated as \(\tx_t = X_t
\o - r_t\) using a time-invariant weight vector \(\o\) with dimensions
\(\brac{m \times 1}\) and a measurement error term \(r_t\). In matrix
form, this reduces to:

\begin{equation} \label{eqn:lsr_b}
y_t = c + X \bko + \epsilon_t
\end{equation}

where \(X := \begin{bmatrix} X_t, & X_{t-1}, & X_{t-2}, & \ldots &
X_{t-F} \end{bmatrix}\) is a \(\brac{1 \times mF}\) vector of observed
variables which contains all the lag powers of the proxy vector \(X_t\),
and \(\b\) has dimensions \(F \times 1\). This paper will henceforth
refer to (\ref{eqn:lsr_b}) a Latent Shock Regression (LSR).

The unconstrained solution for the LSR model is given by:

\begin{subnumcases}{\hspace{-15em}\label{eqn:lsr_solution_full}}
\oh = \left[ \bki' \Sm_{X} \bki \right]^{-1} \bki' \Sm_{Xy}
    \label{eqn:lsr_fp_1} \\[10pt]
\bh = \left[ \oki' \Sm_{X} \oki \right]^{-1} \oki' \Sm_{Xy}
     \label{eqn:lsr_fp_2}
\end{subnumcases}

\begin{equation} \label{eqn:lsr_solution_const}
\hspace{-26em} \ch = \overline{\y} - \overline{\X} \bko
\end{equation}

One can look at the LSR model in two ways. On the one hand, it is an SDM
with \(F\) lags of a single latent factor. On the other hand, it is a
multivariate lead-lag regression with \(F\) lags of \(m\) observed
explanatory variables. The latter interpretation makes LSR directly
comparable to a traditional lead-lag regression model, which reveals one
key difference: In a standard model, each lag of each regressor would be
mapped to a unique regression coefficient, resulting in \(mF\) slope
parameters. Meanwhile, the LSR model only requires \(F\) coefficients
for the vector \(\b\) and \(m\) coefficients for the vector \(\o\),
resulting in \(F + m\) slope parameters. This difference can lead to a
substantial divergence in model complexity when both \(F\) and \(m\) are
large.

The gain in parsimony achieved by LSR models comes from one key
simplifying assumption. In LSR models, all the observed explanatory
variables are treated as facets (``proxies'') of a single unobserved
shock process \(\tx\). As a result, all observed explanatory variables
are assumed to affect the dependent with a shared lag profile: the lag
profile of \(\tx\) given by the \(\brac{F \times 1}\) slope vector
\(\b\). Meanwhile, the weights of the observed variables in \(\tx\) are
time-invariant and given by the \(\brac{m \times 1}\) vector \(\o\).
What would be an \(\brac{mF \times 1}\) coefficient vector in a
traditional lead-lag regression is then replaced by an \(\brac{mF \times
1}\) Kronecker product \(\brac{\b \otimes \o}\), and the slope
coefficient for lag \(\tau\) of predictor \(i\) is given by the product
of the \(i\)'th element of \(\o\) and the \(\tau\)'th element of \(\b\).

This limiting assumption can lead to a reduction in the quality of
in-sample regression fit for LSR models compared to traditional lead-lag
regression models. On the other hand, it can also reduce the risk of
overfitting when the assumption about homogenous lag profiles is
realistic. Furthermore, the structural constraint of the LSR model can
be relaxed by slicing the proxy vector \(X_t\) into blocks and
constructing multiple latent shock processes \(\tx_j\) using a
block-wise Kronecker product \(\bo\). Multivariate LSR (MLSR)
specifications can provide additional flexibility in terms of the number
of unique lag profiles of transmission (and, accordingly, the number of
required regression parameters), serving as a middle ground between pure
(univariate) LSR and traditional lead-lag regression models\footnote{Given a \(\brac{1 \times m}\) vector of observed regressors
\(X_t\), a multivariate LSR model would reduce to a univariate LSR model
when \(X_t\) has one block and to a traditional lead-lag regression
model when \(X_t\) has \(m\) blocks.}.
\subsection{LVMR: A Regression Model Solvable by Canonical Correlation Analysis}
\label{cca}
Section \ref{sdf} argues that the LARX model subsumes Canonical Correlation
Analysis (CCA). To demonstrate this, let us consider a subcategory of
LARX models in which each latent variable enters the equation with
exactly one version (no duplicate weight vectors). We can broadly
categorise these models as Latent Variable Multiple Regression (LVMR)
models because of their resemblance to the class of models examined in
\citep{Burnham-MacGregor-1996}. LVMR models have two simplifying
features compared to the full LARX specification: First, there are no
autoregressive lags and hence no vector \(\p\). Second, because each
latent explanatory variable only appears once, the entire vector \(\b\)
becomes redundant\footnote{Formally, if each \(\tx_j\) only has one version, then the
corresponding block \(\b_j\) of coefficient vector \(\b\) only has one
element. Each term \(\X_j \brac{ \b_j \otimes \o_j}\) then reduces to
\(X_j \o_j \beta_j\) where \(\beta_j\) is a scalar. If there is no
scaling constraint imposed on \(\o_j\), \(\beta_j\) becomes redundant as
it can be ``absorbed'' into the solution for \(\o_j\). If there is a
scaling constraint on \(\o_j\), then it can be imposed by multiplying
the unconstrained solution by the appropriate scalar value after the
fact and \(\beta_j\) would represent that scalar value.}. With these simplifications, the regression
formula reduces to:

\begin{equation}\label{eqn:lvmr_b}
Y \w = c + X \o + \epsilon
\end{equation}

The solution to this problem is given by:

\begin{subnumcases}{\hspace{-27em}\label{eqn:lvmr_solution}}
\wh = \frac{1}{\ry} \Sm_{Y}^{-1} \Sm_{YX} \o
   \label{eqn:lvmr_fp_1} \\[10pt]
\oh = \Sm_{X}^{-1} \Sm_{XY} \w
   \label{eqn:lvmr_fp_2} \\[10pt]
\ry = \frac{ \w' \, \Sm_{YX} \o }{\s2_y}
     \label{eqn:lvmr_fp_3}
\end{subnumcases}

\begin{flalign} \label{eqn:lvmr_solution_const}
& \ch = \overline{\Y} \w - \overline{\X} \o &
\end{flalign}

LVMR models sit at the cusp between LARX, CCA and traditional multiple
linear regression models. Firstly, equation (\ref{eqn:lvmr_fp_2})
represents the least squares solution for a multiple linear regression
between \(\ty\) and the individual observed variables in \(X\)
(semantically, the vector \(\o\) could just as well be called \(\b\),
and whether any blocks of \(\o\) represent LV weight vectors for some
latent variables \(\tx_j\) is a question of interpretation only). As a
corollary, an LVMR model reduces to a standard multiple regression model
when \(\ty\) is non-latent. Secondly, solving (\ref{eqn:lvmr_solution}) is
equivalent to finding the canonical variates for \(Y\) and \(X\), as
previously observed by \citep{Dong-2018}. If we substitute
(\ref{eqn:lvmr_fp_2}) into (\ref{eqn:lvmr_fp_1}) and set \(\s2_y = 1\),
we obtain:

\begin{equation} \label{eqn:cca_fp}
\wh = \frac{\Sm_{Y}^{-1} \Sm_{YX} \Sm_{X}^{-1} \Sm_{XY} \w}
         {\w' \Sm_{Y}^{-1} \Sm_{YX} \Sm_{X}^{-1} \Sm_{XY} \w}
\end{equation}

which is the mathematical formula for CCA.
\subsection{Latent Variable Autoregression (LAR) and Canonical Autocorrelation Analysis (CAA)}
\label{caa}
Another interesting special case of the LARX model is a latent variable
autoregressive model with no exogenous inputs (LAR), which can be viewed
as an SDM counterpart of the traditional autoregressive model. This
section briefly discusses the LAR model and its own special case: a
latent autoregressive model of order LAR(1). With a few simplifying
assumptions, the LAR(1) model reduces to a new type of
eigendecomposition problem which is similar to PCA and CCA.

LAR models can be defined by stripping away the exogenous term from
equation (\ref{eqn:larx_b}), which results in the following formula:

\begin{equation} \label{eqn:clar}
Y \w = c + A \pw + \epsilon
\end{equation}

The fixed point solution to this problem is given by:

\begin{subnumcases}{\hspace{-14em}\label{eqn:clar_solution_full}}
\wh = \frac{1}{\ry} \Sm_Y^{-1} \brac{ \tv_1 - \r_l \1_n }
   \label{eqn:clar_fp_1}  \\[10pt]
\ph = \left[ \wi' \Sm_{A} \wi \right]^{-1} \wi' \Sm_{AY} \w
     \label{eqn:clar_fp_2} \\[10pt]
\ry = \frac{ \left( n \w - l_y \1_n \right)' \tv_1 }
             {n \s2_y - l_y \1_n' \Sm_{Y} \w}
     \label{eqn:clar_fp_3} \\[10pt]
\r_l = \frac{1}{n} \left[
    \1_n' \tv_1 - \ry \, \1_n' \Sm_{Y} \w
\right] \label{eqn:clar_fp_4}
\end{subnumcases}
\begin{flalign*}
& \text{with } \tv_1 = \left[
    \pi' \Sm_{AY} +
    \Sm_{YA} \pi -
    \pi' \Sm_{A} \pi
\right] \w &
\end{flalign*}

Similar methodologies have been examined by papers from other
disciplines\footnote{For example, \citep{Dong-2018} examines a similar type of model
under the name DiCCA (dynamic inner CCA), while \citep{Qin-2021}
proposes a LaVAR (latent vector autoregression) algorithm for achieving
a full canonical decomposition of the latent autoregressive structure in
\(Y\) allowing for interactions. First-order LAR models also bear a
strong resemblance to Min/Max Autocorrelation Factors (MAF) first
introduced in \citep{Switzer-1984} and since popularised in the
geosciences. A thorough comparison between these models is left to
future research.}. The main use case of these models outside of
economics of finance is the decomposition of multivariate sensor data
into time-persistent signals on the one hand, and serially uncorrelated
white noise on the other. In the context of investment management, the
same concept can be applied to derive trend following investment
strategies such as price momentum. Suppose the vector \(Y\) represents
the returns on a collection of investible assets, and \(Y \w_i\)
represents the return on an investment strategy characterised by capital
allocation weights \(\w_i\). The fixed point defined by
(\ref{eqn:clar_solution_full}) would then produce a vector of capital
allocation weights \(\wh\) corresponding to an investment strategy with
the strongest possible price momentum or reversal signal, i.e., the
strongest possible correlation between past returns and future returns
in absolute terms, based on a sample of historical returns \(\Y\).

A more refined result is obtained by examining the simplest type of LAR
problem, namely, a first-order autoregressive model of the form \(\ty_t
= c + \phi \ty_{t-1} + e_t\) with \(\ty_t = Y \w + r_y\) and \(\ty_{t-1}
= A \w + r_a\). In this case \(\p\) reduces to a scalar \(\phi\) and the
unconstrained solution becomes:

\begin{subnumcases}{\hspace{-19em}\label{eqn:lar1_solution_full}}
\wh = \frac{1}{\ry} \Sm_Y^{-1} \phi
      \left[\Sm_{AY} + \Sm_{YA} - \phi \Sm_{A} \right] \w
   \label{eqn:lar1_fp_1}  \\[10pt]
\phih = \frac{\w' \Sm_{AY} \w}{\w' \Sm_{A} \w}  
     \label{eqn:lar1_fp_3} \\[10pt]
\ry = \frac{\phi }{\s2_y} \w' \left[
    \Sm_{AY} +
    \Sm_{YA} -
    \phi \Sm_{A}
\right] \w
     \label{eqn:lar1_fp_2}
\end{subnumcases}

If we further assume that \(Y\) is covariance stationary, i.e., \(\Sm_Y
= \Sm_A\), \(\phih\) reduces to \(\frac{\w' \Sm_{AY} \w}{\s2_y}\) and
\(\ry\) can be refactored as:

\begin{flalign*}
& \ry = \frac{\phi }{\s2_y} \w' \left[
          \Sm_{AY} + \Sm_{YA} - \phi \Sm_{A}
        \right] \w
      = \frac{\phi }{\s2_y} \left[
            \w' \Sm_{AY} \w + \w' \Sm_{YA} \w - \phi \w' \Sm_{A} \w
        \right] & \\
& \quad = \phi \left[
            \frac{\w' \Sm_{AY} \w}{\s2_y} + \frac{\w' \Sm_{YA} \w}{\s2_y}
            - \phi \frac{\w' \Sm_{A} \w}{\s2_y}
          \right]
        = \phi \left[ \phi + \phi - \phi \frac{\s2_y} {\s2_y} \right]
        = \phi^2 &
\end{flalign*}

The solution for \(\w\) then becomes:

\begin{equation} \label{eqn:w_eigenvec_lar1}
\phih \wh = \frac{1}{2} \Sm_Y^{-1} \left[ \Sm_{AY} + \Sm_{YA} \right] \w 
\end{equation}

or, equivalently:

\begin{equation} \label{eqn:w_eigenvec_lar1_b}
\phih \wh = \frac{1}{2} \left[ \Sm_A^{-1} \Sm_{AY} + \Sm_Y^{-1} \Sm_{YA} \right] \w
\end{equation}

In other words, if the proxy vector \(Y\) is covariance stationary, a
LAR(1) model specified by (\ref{eqn:lar1_solution_full}) can be estimated by
the eigendecomposition of the matrix \(\frac{1}{2} \Sm_Y^{-1}
\brac{\Sm_{AY} + \Sm_{YA}}\), with each \(\wh\) being its eigenvector
and \(\phih\) the matching eigenvalue. An appropriate name for this
methodology would be ``Canonical Autocorrelation Analysis''\footnote{The term Canonical Autocorrelation Analysis was previously used
in \citep{Chen-2015} to describe a different type of statistical model.
However, given that the paper is in a different field, applying the same
term to equation (\ref{eqn:w_eigenvec_lar1}) should not invite any ambiguity.} (CAA)
because it produces a full canonical decomposition of the directions of
first-order autocorrelation in a group of covariance-stationary time
series variables \(Y\).

All components of equation (\ref{eqn:w_eigenvec_lar1}) have reasonably
intuitive interpretations. The matrices \(\Sm_A^{-1} \Sm_{AY}\) and
\(\Sm_Y^{-1} \Sm_{YA}\) are basically just matrices of least squares
regression coefficients: The first (second, third, \ldots{}) column of
\(\Sm_A^{-1} \Sm_{AY}\) is a vector of OLS coefficients from a
regression of the first (second, third, \ldots{}) variable in \(Y\) on all
the variables in \(A\). For the matrix \(\Sm_Y^{-1} \Sm_{YA}\) the roles
are reversed, i.e., its columns represent the OLS coefficient vectors
obtained from regressing each variable in \(A\) on all variables in
\(Y\). The coefficient \(\phi\) is both the autocorrelation coefficient
of \(\ty\) and the eigenvalue of \(\frac{1}{2} \Sm_Y^{-1} \brac{
\Sm_{AY} + \Sm_{YA}}\), which means that the first (last) eigenvector
produces the strongest (weakest) absolute autocorrelation signal.

In an investment management context, equation (\ref{eqn:w_eigenvec_lar1}) can
be used to extract momentum and reversal patterns from a collection of
investible assets. Let \(Y\) and \(A\) represent the returns on an
investment opportunity set at times \(t\) and \(t-1\), respectively.
Each eigenvector \(\wh_i\) of the matrix \(\Sm_Y^{-1} \brac{\Sm_{AY} +
\Sm_{YA}}\) would be a vector of capital allocation weights. The
investment strategies produced by all \(\wh_i\) would capture all the
distinct directions of autocorrelation in \(Y\). The dominant
eigenvectors would produce investment strategies with the strongest
momentum or reversal signals, while the last eigenvectors would produce
strategies whose returns are closest to a random walk. Such a collection
of investment portfolios could be useful for a range of tasks, including
identifying market inefficiencies, predicting asset returns, and pricing
systematic investment risks.
\section{Empirical Application: Stock Markets and Economic Activity in the US}
\label{empirical_results}
As an example of how the LARX model can be used in the real world, let
us briefly examine the relationship between equity market performance
and real economic activity in the United States. A good starting point
for this analysis is provided by \citep{Ball-2021} who find that
de-trended levels of the S\&P 500 index have in-sample predictive power
over the de-trended levels of real US GDP.

The relationship between stock returns and real economic activity has a
strong foundation in economic theory. For example, the Consumption-based
Capital Asset Pricing Model (CCAPM) (\citep{Lucas-1978,Breeden-1979})
and the Investment CAPM (ICAPM) (see \citep{Zhang-2017} and the
references therein) put forward consumption and investment,
respectively, as the main sources of priced risk in the equity market.
With both consumption and investment also being key components in the
expenditure model of GDP, the relationship between stock returns and GDP
growth follows naturally.

In reality, however, market aggregates and macroeconomic aggregates are
designed to measure different things. Companies in the S\&P 500 are
weighted based on their market capitalisation rather than their relative
contributions to the real economy. US GDP is designed to measure
broad-based economic activity in the US -- not just that of the Fortune
500 companies. These differences suggest that the true strength of the
relationship between stock performance and real economic output would be
underestimated by a model linking the S\&P 500 to real US GDP.

The LARX model can be used to address this discrepancy. Both the S\&P 500
and US GDP are composite measures, which means they can be broken down
into their constituent parts and reassembled into latent measures of
market performance and economic growth, respectively. We can use the
five expenditure components of US GDP as proxy variables for a
supervised diffusion index (SDI) of US real economic activity (a ``real
activity SDI''), and 10 GICS\footnote{\href{https://www.msci.com/indexes/index-resources/gics}{Global Industry Classification Standard}} level 1 sector sub-indices\footnote{As of 2016, listed real estate (RE) was added as the eleventh
GICS level 1 sector of the S\&P 500. The RE sector is excluded from this
study for two reasons: Firstly, its data history only starts in Q4 2001
and would reduce the sample size from 138 quarterly observations to 90.
Secondly, the RE sector only has a 2.25\% weight in the S\&P 500 as of
April 2025 -- the second lowest weight after materials at 1.99\%.} of
the S\&P 500 as proxy variables for an SDI of market growth expectations
(a ``market SDI''). The component weights of the SDI measures can then
be compared to the component weights in the official aggregates to
explore whether the relationship between the S\&P 500 and real US GDP is
distorted by a misalignment in sector composition, the relative
importance of different sources of demand, or both.
\subsection{Data and Methodology}
\label{data_and_methodology}
Historical data for US GDP and its expenditure components are retrieved
from the Economic Database of the Federal Reserve Bank of St. Louis
(``FRED''). Historical index levels for the S\&P 500 and its GICS level 1
sector constituents are retrieved from Investing.com. A full data
reference can be found in Table \ref{table:data_reference}.

Empirical results from \citep{Ball-2021} suggest that the
best-performing model for real US GDP at time \(t\) contains the S\&P 500
at times \(t\) to \(t-3\) and two autoregressive lags, obtaining a
reported adjusted R-squared of 66.61\% for a quarterly data sample
between Q1 1999 and Q4 2020. Let us take the same specification as the
baseline but make three noteworthy changes to the experiment design:
First, model performance is going to be measured out of sample, rather
than in-sample. Second, log-percent changes will be used instead of
de-trended levels, which is a more common approach to measuring both
economic activity and stock market performance. Third, the sample period
will cover Q4 1989 to Q3 2025, capturing the full available data history
as of the time of writing\footnote{The main constraint on history length is the GICS classification
for the 10 original S\&P 500 sectors which was only introduced in 1999
with the data back-dated to 1989.}. Revised estimates will be used for all
economic aggregates following the tentative finding in \citep{Ball-2021}
that the link is stronger between equity performance and revised GDP
numbers as opposed to point-in-time (``vintage'') releases. A total of
four regression models will be compared:

\begin{subequations} \label{eqn:regression_models}
Baseline OLS/ARX model: Real GDP growth \(g\) vs S\&P 500 returns \(r\):
\begin{equation} \tag{\ref*{eqn:regression_models}}
g_t = c + \sum_{\tau = 1}^2 \phi_{t-\tau} g_{t-\tau} +
      \sum_{\tau = 0}^3 \beta_{t-\tau} r_{t-\tau} + e
\end{equation}
LARX model a): Latent explanatory: Real GDP growth vs market SDI returns \(\tr\):
\begin{equation}\label{eqn:regression_g_nonlatent}
g_t = c + \sum_{\tau = 1}^2 \phi_{t-\tau} g_{t-\tau} +
      \sum_{\tau = 0}^3 \beta_{t-\tau} \tr_{t-\tau} + e
\end{equation}
LARX model b): Latent dependent: Change in the real activity SDI \(\tg\) vs S\&P 500 returns:
\begin{equation}\label{eqn:regression_r_nonlatent}
\tg_t = c + \sum_{\tau = 1}^2 \phi_{t-\tau} \tg_{t-\tau} +
      \sum_{\tau = 0}^3 \beta_{t-\tau} r_{t-\tau} + e
\end{equation}   
LARX model c): Both latent: Change in the real activity SDI vs market SDI returns:
\begin{equation}\label{eqn:regression_all_latent}
\tg_t = c + \sum_{\tau = 1}^2 \phi_{t-\tau} \tg_{t-\tau} +
      \sum_{\tau = 0}^3 \beta_{t-\tau} \tr_{t-\tau} + e
\end{equation}
\end{subequations}

For a cleaner comparison, variance constraints are imposed on the market
SDI and the real activity SDI to mimic the full-sample variance of the
S\&P 500 and real GDP growth, respectively. Although this information is
only available with the benefit of hindsight, the choice of variance
constraint does not affect the performance of the LARX model as it is
scale invariant. All regressions use exponentially decaying sample
weights with a half-life of 10 years, in order to allow for some drift
in the regression parameters. The COVID lockdown period of Q2 and Q3
2020 is treated as a statistical outlier and removed from the data
sample\footnote{US GDP shows a contraction of 8.2\% in Q2 2020 followed by a 7.5\%
rebound in Q3 2020 -- a -13.8 sigma event and a 12.6 sigma event,
respectively, based on the standard deviation of quarterly US GDP growth
excluding these two quarters.}. A minimum of 40 degrees of freedom is set as a requirement
for producing a forecast, which corresponds to 10 years of quarterly
data on top of one data point lost to each estimated coefficient,
including Lagrange multipliers. An additional three data points are lost
to the lag operator and one to the percent change calculation.
Ultimately, forecast coverage starts in Q3 2002 for the baseline model
(longest) and in Q3 2006 for the model with all latent variables
(shortest).
\subsection{Out-of-Sample Forecasting Performance}
\label{sec:org885f59d}

Figure \ref{NEW_fig-fc_ols_vs_clarx} plots the rolling out-of-sample (OOS)
predictions made by the baseline model (top left), as well as LARX
models a) (top right), b) (bottom left), and c) (bottom right). Each
plot overlays the actual values of the dependent (real GDP growth in the
top row, the real activity SDI in the bottom row), as well as the naïve
forecast for the dependent based a rolling sample mean (a.k.a. the
``benchmark''). The grey text boxes report each model's OOS R-squared
statistic as defined in \citep{Campbell-2007}.

\begin{figure}[ht]
    \includegraphics[width=\textwidth]{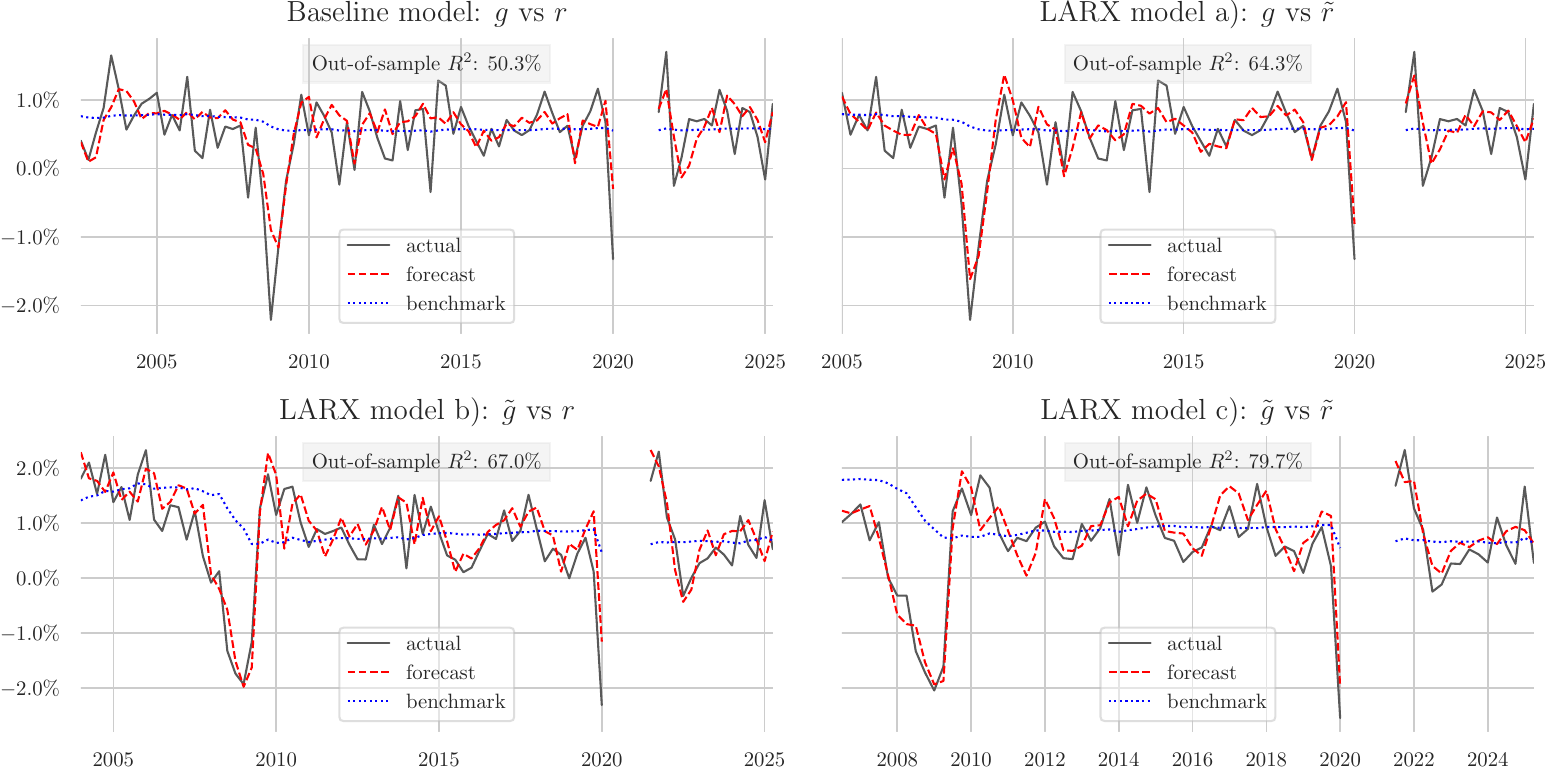}
    \caption{Out-of-sample forecasting performance} \label{NEW_fig-fc_ols_vs_clarx}
    \floatfoot{ \footnotesize %
      Note: Rolling out-of-sample forecasting performance of the four
      models specified by equations
      (\ref{eqn:regression_models})-(\ref{eqn:regression_all_latent}).
      The target variance for the real activity SDI is set to the
      full-sample variance of quarterly US real GDP growth (approx.
      0.000124). The variance of the market SDI is set to the
      full-sample variance of the S\&P 500 quarterly log-return
      (approx. 0.006348).
    }
\end{figure}

The baseline model from \citep{Ball-2021} (top left) does well despite
the changes to the experiment design. It achieves an OOS R-squared of
50.3\% -- reasonably close to the in-sample adjusted R-squared of 66.61\%
reported in the original study.

All three LARX specifications further improve on the baseline model.
LARX model a) (market SDI and real GDP growth -- top right) yields an
OOS R-squared of 63.9\%. LARX model b) (the S\&P 500 and the real activity
SDI -- bottom left) produces an OOS R-squared of 67.0\%. LARX model c)
(marked SDI and real activity SDI -- bottom right) performs best with an
OOS R-squared of 79.7\%.

For completeness, it must be noted that the dependent variable in models
b) and c) is different from the dependent variable in models a) and b).
However, the comparison is still meaningful for at least three reasons.
First, the variance of the real activity SDI matches the variance of
real US GDP growth by design. Second, the goal of the study is to assess
the importance of a common unobserved factor in both the stock market
and the real economy, which is a bi-directional problem requiring
optimisation on both ends. Third, rearranging equations
(\ref{eqn:regression_models})-(\ref{eqn:regression_all_latent}) to use
the stock market return as the dependent yields a largely similar set of
results, as reported in \ref{empirical_results_rev}.
\subsection{Explaining the Outperformance of the LARX Models}
\label{sec:orgef6b772}

We can dive deeper into the empirical results for additional insight.
Overall, LARX models a) and b) both outperform the baseline model by
roughly the same amount, while LARX model c) yields approximately the
same improvement as LARX models a) and b) combined. This suggests three
things: First, the sector composition of the market SDI better reflects
the sector composition of the real economy. Second, the expenditure
composition of the real activity SDI better reflects the main sources of
demand for the output of Fortune 500 companies. Third, sector
composition and expenditure composition represent two equally important
yet distinct sources of discrepancy between the stock market and the
real economy, as fixing one has little to no bearing on the extent of
other.

To gain a better understanding of the underlying dynamics, let us look
at the weight vectors of the SDI measures obtained from LARX model c)
(the best-performing specification) side by side with their non-latent
counterparts.

\begin{figure}[ht]
    \includegraphics[width=\textwidth]{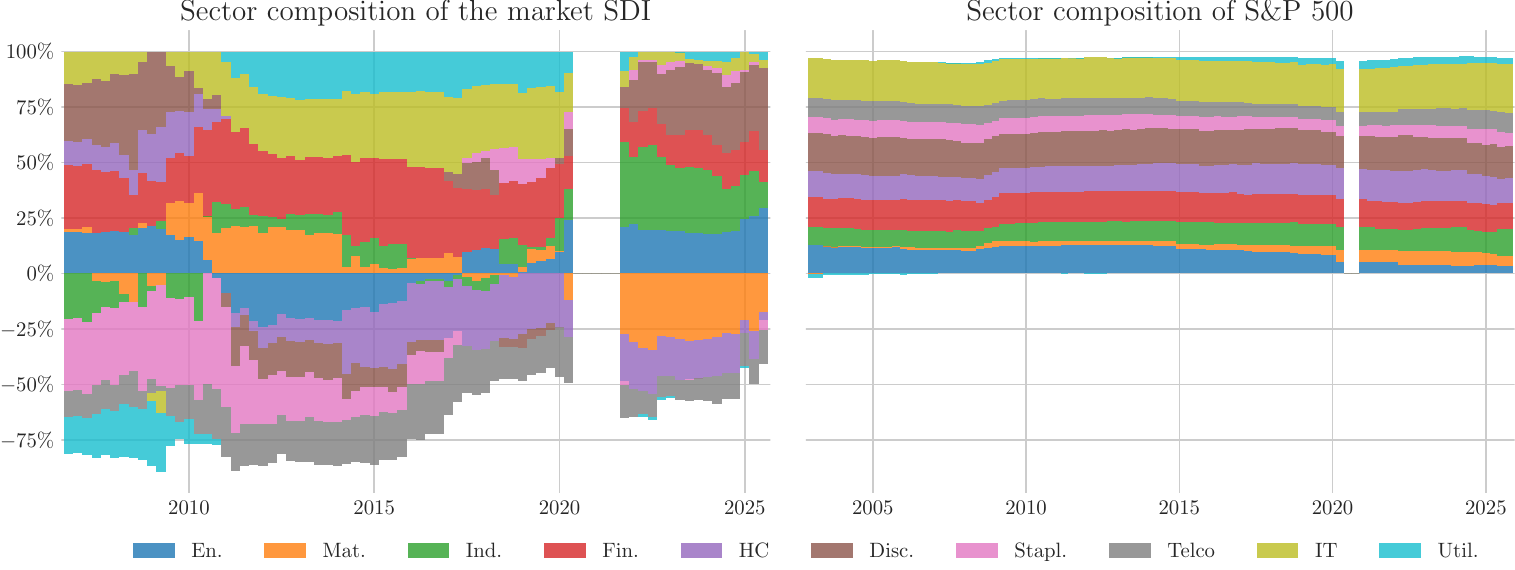}
    \caption{Sector composition: Latent Measure of Market Growth Expectations vs the S\&P 500}
    \label{NEW_fig-weight_evolution_spx}
\end{figure}

Figure \ref{NEW_fig-weight_evolution_spx} compares the sector weights of the
market SDI (left) with the sector weights of the S\&P 500 (right)
approximated using a rolling regression of S\&P 500 returns on the
coincident returns of the sector sub-indices. The weight vector of the
market SDI is scaled so that the positive weights add up to 100\%.

Here, two observations jump out. First, sector rotations, a.k.a.
relative sector performance, comprise between 40\% and 80\% of the market
SDI at any given point in time. This suggests that sector rotations are
about as useful for forecasting real economic activity as is the overall
direction of the equity market. Sector rotations have been previously
examined in the investment management literature for the purposes of
capturing systematic risk premia (see \citep{Lamponi-2014}). A link
between sector rotations and the business cycle has been examined for
the purposes of constructing systematic investment strategies by
\citep{Molchanov-2024}, albeit with limited success. To the author's
best knowledge, no previous study demonstrates that sector rotations
have out-of-sample predictive power with respect to real economic
growth.

Second, the composition of the market SDI fluctuates quite strongly over
time, and many of these fluctuations seem to correspond to major
economic trends or events. For example, the performance of the energy
(``En.'') sector changes its economic meaning after the introduction of
major green energy initiatives by the Obama administration in 2009, and
again around the commodity super-cycle of 2015-2016. Similarly, the
healthcare sector (``HC'') has a positive weight until the introduction
of the Affordable Care Act (a.k.a. ``Obamacare'') in 2010 and a negative
weight thereafter.

\begin{figure}[ht]
    \includegraphics[width=\textwidth]{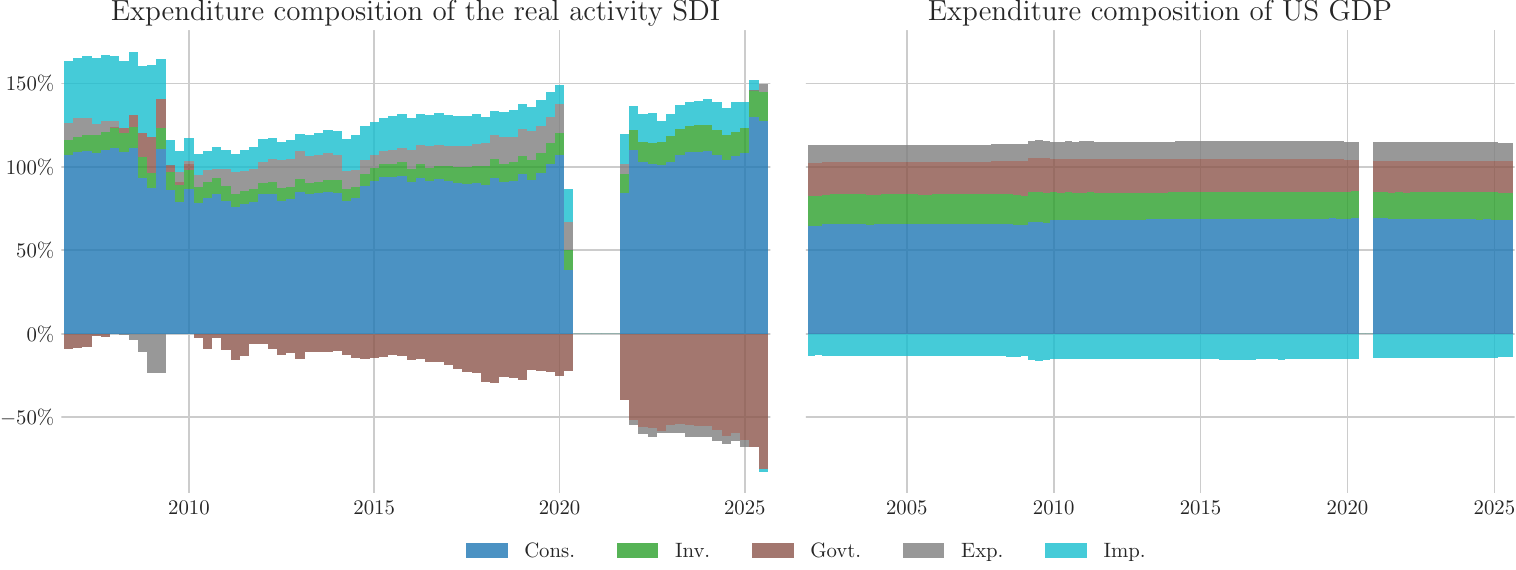}
    \caption{Expenditure composition: Latent Economic Activity Measure vs US GDP}
    \label{NEW_fig-weight_evolution_gdp}
\end{figure}

Figure \ref{NEW_fig-weight_evolution_gdp} plots the evolution of the expenditure
weights in the real activity SDI (left) against the evolution of the
expenditure weights in real US GDP (right).

Both measures assign the largest weight to consumer spending
(``Cons.''), which is an expected result for a consumption-led economy
like the United States. The role of private investment (``Inv.'') is
smaller and largely similar for both aggregates.

The biggest differences lie in the role of government spending
(``Govt.'') and the trade balance (``Imp.'', ``Exp.''). In the real
activity SDI, stronger government spending is associated with weaker
equity market performance, and vice versa, outside of the immediate
aftermath of the 2008 crisis. In other words, equity investors seem to
welcome fiscal expansion in a crisis, but view it as a negative in
normal times, a.k.a. the Keynesian notion of fiscal policy as a
counter-cyclical buffer. The effect becomes more pronounced after the
outbreak of COVID-19, perhaps owing to a ballooning US fiscal deficit in
2020 and a sharp rise in interest rates since 2021.

Lastly, the interplay between the trade balance and the equity market
seems to be much more complex than the role of imports and exports in
the national accounts. \citep{Reinbold-2019} offers a brief primer on
the history of the US trade balance across the value chain at different
stages of the country's industrialisation. From an accounting
perspective, imports subtract from GDP while exports add to it, both
with a multiplier of 1. However, today many large US companies have
manufacturing facilities in other countries, which means that their
products need to be imported before they are sold domestically, but they
don't need to exported to be sold abroad. As a result, the link between
stock returns and exports is weakened, whereas the link between stock
returns and imports becomes positive inasmuch as imports become a
leading indicator of domestic sales. The real activity SDI captures this
relationship by assigning a time-varying weight to US exports and a
consistently positive weight to US imports, with the exception of the
latest quarter in the sample (3Q 2025) characterised by major swings in
US tariff policy.
\subsection{Questions for Future Research}
\label{sec:org57fb151}

A more thorough examination of the relationship between stock returns
and real economic activity is beyond the scope of this paper; however,
several avenues of further investigation may be of interest to future
research. Three examples are listed below.

First, further insight may be gained by constructing the market SDI and
the real activity SDI at higher levels of granularity. The market SDI
can be estimated using GICS level 2 indices or even single stocks, while
the real activity SDI can go all the way down to the itemised national
accounts. Of course, the benefits of a more granular approach should be
weighed against a higher risk of overfitting.

Second, this paper's findings are not examined from an asset pricing
standpoint. No attempts are made to test whether the market SDI produced
by the LARX model is a valid asset pricing factor in the CCAPM or ICAPM
sense, a profitable investment strategy, or a consistent leading signal
for business cycle rotations.

Third, only the case of the United States is examined. It may be
informative to perform similar studies for countries with other economic
models and sector compositions, in order to see whether the relationship
between stock returns and domestic economic activity is a universal
phenomenon.
\section{Concluding remarks}
\label{concluding_remarks}
This paper proposes a new methodological framework for estimating latent
variable models called Supervised Diffusion (SDF). A fixed point
solution is derived for a new supervised diffusion model (SDM) called
LARX: a superset of the traditional autoregressive model with exogenous
inputs (ARX) in which any or all variables can be latent. The
derivations of the LARX model result in a minor contribution to the
field of matrix calculus: A block-wise direct sum operator is introduced
and applied to solve a class of Lagrangian optimisation problems with
interactions between multiple coefficient vectors in the presence of
piecemeal constraints. Several special cases of the LARX models are also
examined in more detail, including a parsimonious lead-lag regression
model called Latent Shock Regression (LSR), as well as a new canonical
decomposition technique called Canonical Autocorrelation Analysis (CAA).

In the empirical section, the LARX model is used to re-examine the link
between stock market performance and real economic activity in the
United States. The LARX model attains an out-of-sample (OOS) R-squared
of 79.7\%, compared to an OOS R-squared of 50.3\% for the baseline OLS
specification. The LARX model also provides novel insights about the
interplay between stock returns and the real economy, corroborating a
few notions which were previously best described as ``rules of thumb''.
For example, new evidence is found in support of the importance of
sector rotations in predicting economic growth, as well as the complex
and multifaceted roles of fiscal policy and international trade in the
economic value chain.

SDMs like LARX have many potential use cases in economics and finance.
For example, they can be used to obtain better estimates for variables
whose main goal is to track other processes, like diffusion indices of
business activity (e.g., \citep{Owens-2005}), surveys of consumer
sentiment (e.g., \citep{Curtin-2000}), or composite indicators of
financial stress (e.g., \citep{CISS-2012}). Supervised diffusion indices
can also replace various heuristic investment techniques, such as
``buying past winners and selling past losers'' as a means of capturing
asset price momentum as in \citep{Jegadeesh-1993}.

In the scientific process, SDMs have a place in research settings
characterised by noisy and unreliable empirical data and little scope
for controlled experiments -- a common problem in macroeconomics. SDMs
can enable the researcher to place relatively more trust in the research
hypothesis by expanding the search for its supporting evidence to
arbitrary linear combinations over the observed data series. Formally
speaking, the use SDMs reduces the risk of Type II error at the expense
of a somewhat higher chance of Type I error.

Sections \ref{sdf}-\ref{larx_special_cases} leave much scope for further
methodological work. For example, the topics of statistical significance
and feature selection\footnote{In the specific case of CCA, \citep{Bagozzi-1981} and
\citep{Ahn-2018} offer good starting points in the analysis of
statistical significance and feature selection, respectively.} were only briefly touched upon in Section
\ref{larx_ols_interpretation}. Furthermore, the LARX methodology lends itself
well to a number of adjustments which are common practice with the
traditional ARX model, including moving average errors (MA), conditional
heteroskedasticity (GARCH), seasonality patterns, various forms of
coefficient regularisation (e.g., LASSO, Ridge, Elastic Net - see, for
example, \citep{Vinod-1976}), and various covariance adjustment
techniques such as Generalised Least Squares (e.g., \citep{Aitken-1936})
and portfolio-style covariance shrinkage (e.g., see
\citep{Ledoit-Wolf-2022}).

Non-linear SDMs may warrant a closer look as well. For example, in the
context of the asset pricing theory, the LARX model can be used for
estimating risk factors such as price momentum (\citep{Jegadeesh-1993}),
earnings momentum (\citep{Barth-1999}) and company size
(\citep{Fama-French-1993}), which are governed by linear relationships.
However, it cannot be used as easily to estimate factors such as value
(e.g., \citep{Asness-2013,Fama-French-1993,Shiller-1988}) or quality
(e.g., see \citep{Hsu-2017}). Valuations are, generally speaking, ratios
of prices to fundamentals, e.g., a portfolio's price-to-earnings (PE)
valuation can be calculated as \(\frac{P \w}{EPS \w}\), where \(P\) and
\(EPS\) are vectors of company share prices and earnings per share, and
\(\w\) is a vector of capital allocation weights. A value factor SDI
would then be governed by a functional relationship similar to:

\begin{equation} \label{eqn:reg_value_factor}
Y_t \w = c + \frac{P_{t-1} \w}{EPS_{t-1} \w} + \epsilon_t
\end{equation}

where \(Y\) is a vector of constituent returns, \(c\) is the intercept,
\(\epsilon\) is the error term, and \(t\) is a time subscript. The
quality factor, on the other hand, can be based on earnings volatility
(e.g., see \citep{Dichev-2009}), which would necessitate a quadratic
SDM.

To summarise, SDMs such as LARX can be viewed as a rather natural
evolution of traditional regression analysis. At the same time, an
effective application of these models requires a slight paradigm shift
on the part of the researcher, because their input variables are
designed to change shape according to the functional relationship at
hand. This feature may not be universally useful, but it does create an
opportunity to re-examine a number of existing models and discover new
relationships in economics and finance, as well as in a number of other
data-intensive fields.

\begin{appendix}
\section{Commutativity of the blockwise direct sum operator}
\label{bds_vs_mmul}
\renewcommand{\theproposition}{\Alph{section}\arabic{proposition}}

\begin{proposition} \label{prop:bds_vs_mmul}
Let \(S\) be a set of all matrix sequences of length \(k\), and let the
sequence \(\seq{\A} \equiv \seq{\A_i | 1 \leq i \leq k}\) be an element
in \(S\). Let the matrix \(\Adsum\) represent the direct sum over the
elements in \(\seq{\A}\). For any function \(f: S \rightarrow S\), if
\(f \brac{\seq{\A}}\) can be expressed as a sequence \( \seq{\M_i \A_i}
\equiv \seq{\M_i \A_i | 1 \leq i \leq k}\) for some arbitrary sequence
of matrices \(\seq{\M_i | 1 \leq i \leq k}\), then \( f
\left(\seq{\A} \right)^{\oplus}_v = \left[ f \left( \seq{\Adsum_v} \right)
\right]_v \). If \( f \left(\seq{\A}\right) \) can be expressed as a
sequence \(\seq{\A_i \M_i} \equiv \seq{\A_i \M_i | 1 \leq i \leq k}\)
for an arbitrary sequence of matrices \(\seq{\M_i | 1 \leq i \leq k}\),
then \( f \brac{ \seq{\A} }^{\oplus}_h = \left[ f \brac{ \seq{\Adsum_h}
} \right]_h \). \end{proposition}

\begin{proof}
For a block matrix \(\A\), take \(\seq{\A}\) to denote the sequence of
the blocks in \(\A\). For a sequence of matrices \(\seq{\A}\), denote
its \(i\)'th element by \(\seq{\A}_i\). For the case of
left-multiplication we then have:

\begin{flalign*}
& f  \brac{ \seq{\A} }^{\oplus}_v
  = \seq{\M_i \A_i}^{\oplus}_v
  = \brac{ \begin{array}{cccc}
      \M_1 \A_1, & \0, & \cdots & \0 \\ \hline
      \0, & \M_2 \A_2, & \cdots & \0 \\ \hline
      \vdots & \vdots & \ddots & \vdots \\ \hline
      \0, & \0, & \cdots & \M_k \A_k
    \end{array} } & \\[10pt]
& \left[ f \brac{ \seq{\Adsum_v} } \right]_v
  = \left[ \seq{\M_i \seq{\Adsum_v}_i | 1 \leq i \leq k} \right]_v
  = \brac{ \begin{array}{c}
      \M_1 \seq{\Adsum_v}_1 \\ \hline
      \M_2 \seq{\Adsum_v}_2 \\ \hline
      \vdots \\ \hline
      \M_k \seq{\Adsum_v}_k
    \end{array} } = & \\[10pt]
& \hspace{1em} = \brac{ \begin{array}{c}
      \M_1 \begin{pmatrix} \A_1, & \0, & \cdots, & \0 \end{pmatrix} \\ \hline
      \M_2 \begin{pmatrix} \0, & \A_2, & \cdots, & \0 \end{pmatrix} \\ \hline
      \vdots \\ \hline
      \M_k \begin{pmatrix} \0, & \0, & \cdots, & \A_k \end{pmatrix}
    \end{array} }
  = \brac{ \begin{array}{cccc}
      \M_1 \A_1, & \0, & \cdots & \0 \\ \hline
      \0, & \M_2 \A_2, & \cdots & \0 \\ \hline
      \vdots & \vdots & \ddots & \vdots \\ \hline
      \0, & \0, & \cdots & \M_k \A_k
    \end{array} } = f \brac{\seq{\A}}^{\oplus}_v &
\end{flalign*}

For the case of right-multiplication we have:

\begin{flalign*}
& f \brac{\seq{\A}}^{\oplus}_h
  = \seq{\A_i \M_i}^{\oplus}_h
  = \brac{ \begin{array}{c|c|c|c}
      \A_1 \M_1 & \0 & \cdots & \0 \\
      \0 & \A_2 \M_2 & \cdots & \0 \\
      \vdots & \vdots & \ddots & \vdots \\
      \0 & \0 & \cdots & \A_k \M_k \\
    \end{array} } & \\[10pt]
& \left[ f \brac{\seq{\Adsum_h}} \right]_h
  = \left[ \seq{\seq{\Adsum_h}_i \M_i | 1 \leq i \leq k} \right]_v
  = \left[ \begin{array}{c|c|c|c}
      \seq{\Adsum_v}_1 \M_1 &
      \seq{\Adsum_v}_2 \M_2 &
      \cdots &
      \seq{\Adsum_v}_k \M_k
    \end{array} \right] = & \\[10pt]
& \hspace{1em} = \brac{ \begin{array}{c|c|c|c}
      \begin{bmatrix} \A_1 \\ \0 \\ \vdots \\ \0 \end{bmatrix} \M_1 &
      \begin{bmatrix} \0 \\ \A_2 \\ \vdots \\ \0 \end{bmatrix} \M_2 &
      \cdots &
      \begin{bmatrix} \0 \\ \0 \\ \vdots \\ \A_k \end{bmatrix} \M_k
    \end{array} }
  = \brac{ \begin{array}{c|c|c|c}
      \A_1 \M_1 & \0 & \cdots & \0 \\
      \0 & \A_2 \M_2 & \cdots & \0 \\
      \vdots & \vdots & \ddots & \vdots \\
      \0 & \0 & \cdots & \A_k \M_k \\
    \end{array} } = & \\[10pt]
& \hspace{1em} = f  \brac{ \seq{\A} }^{\oplus}_h &
\end{flalign*}
\end{proof}
\section{Blockwise Kronecker Product Factorisation for Vectors}
\label{block_kron_factorisation}

\begin{proposition} \label{prop:block_kron_factorisation}
Let \(\a\) and \(\vb\) be two column vectors, each comprised of \(k\)
row blocks of arbitrary lengths. The blockwise Kronecker product \(\a
\odot \vb\) can be factorised as \(\a \odot \vb = \brac{\a \odot
I_{\vb}} \vb = \brac{ I_{\a} \odot \vb } \a\), where \(I_{\vb}\) and
\(I_{\a}\) are identity matrices with the same number of rows and row
block structure as \(\vb\) and \(\a\), respectively.
\end{proposition}

\begin{proof}
Let the vector \(\a\) have dimensions \(M \times 1\) and \(\seq{\a}
\equiv \seq{\underset{m_i \times 1}{\a_i} | 1 \leq i \leq k}\) be the
sequence of vectors which represent the row blocks in \(\a\) such that
\(\sum_{i = 1}^k{m_i} = M\). Similarly, let the vector \(\vb\) have
dimensions \(N \times 1\) and the sequence of vectors \(\seq{\vb} \equiv
\seq{\underset{n_i \times 1}{\vb_i} | 1 \leq i \leq k}\) represent the
row blocks in \(\vb\) such that \(\sum_{i = 1}^k{n_i} = N\).

The blockwise Kronecker product \(\a \odot \vb\) can then be defined as:

\begin{flalign*}
& \a \odot \vb = \brac{
    \begin{array}{c}
      \a_1 \otimes \vb_1 \\ \hline
      \a_2 \otimes \vb_2 \\ \hline
      \vdots \\ \hline
      \a_k \otimes \vb_k
    \end{array}
} &
\end{flalign*}

Note that by the properties of the Kronecker product the following holds
for any two matrices \(\A\) and \(\B\):

\begin{flalign} \label{eqn:kron_factorisation_matrices}
& \underset{m \times p}{\A} \otimes \underset{n \times q}{\B}
  = \brac{\A \otimes I_n} \brac{I_p \otimes \B}
  = \brac{I_m \otimes \B} \brac{\A \otimes I_q} &
\end{flalign}

In the special case of a Kronecker product between two vectors, \(p\)
and \(q\) reduce to \(1\) and the identity matrices \(I_p\) and \(I_q\)
become \(1\) by association. As a result, for any given Kronecker
product \(\a_i \otimes \vb_i\), the following holds:

\begin{flalign} \label{eqn:kron_factorisation_vectors}
& \underset{m_i \times 1}{\a_i} \otimes \underset{n_i \times 1}{\vb_i}
  = \brac{\a_i \otimes I_p} \vb_i
  = \brac{I_m \otimes \vb_i} \a_i &
\end{flalign}

This allows us to rewrite the blockwise Kronecker product \(\a \odot
\vb\) in two alternative ways:

\begin{subequations}
\begin{flalign}
& \a \odot \vb = \brac{
    \begin{array}{c}
      \brac{\a_1 \otimes I_{n_1}} \vb_1 \\ \hline
      \brac{\a_2 \otimes I_{n_2}} \vb_2 \\ \hline
      \vdots \\ \hline
      \brac{\a_k \otimes I_{n_k}} \vb_k
    \end{array}
} & \label{eqn:bk_deriv_1a} \\[15pt]
& \a \odot \vb = \brac{
    \begin{array}{c}
      \brac{I_{m_1} \otimes \vb_1} \a_1 \\ \hline
      \brac{I_{m_2} \otimes \vb_2} \a_2 \\ \hline
      \vdots \\ \hline
      \brac{I_{m_k} \otimes \vb_k} \a_k
    \end{array}
} & \label{eqn:bk_deriv_1b}
\end{flalign}
\end{subequations}

Define a sequence of matrices \(\seq{\a_i \otimes I_{n_i} | 1 \leq i
\leq k} \equiv \seq{\a_i \otimes I_{n_i}}\), and another
sequence \(\seq{I_{m_i} \otimes \vb_i | 1 \leq i \leq k} \equiv
\seq{I_{m_i} \otimes \vb_i} \). We can then rewrite
(\ref{eqn:bk_deriv_1a}) and (\ref{eqn:bk_deriv_1b}) using the blockwise
direct sum operator introduced in Section \ref{blockwise_direct_sum}:

\begin{subequations}
\begin{flalign}
& \a \odot \vb = \seq{\a_i \otimes I_{n_i}}^{\oplus} \vb & \label{eqn:bk_deriv_2a} \\[10pt]
& \a \odot \vb = \seq{I_{m_i} \otimes \vb_i}^{\oplus} \a & \label{eqn:bk_deriv_2b}
\end{flalign}
\end{subequations}

It then remains to show that \(\seq{\a_i \otimes I_{n_i}}^{\oplus}\) and
\(\seq{I_{m_i} \otimes \vb_i}^{\oplus}\) can be written as \( \a \odot
I_{\vb} \) and \( I_{\a} \odot \vb \), respectively. This can be done
with the help of Proposition \ref{prop:bds_vs_mmul}.

First of all, consider the sequences of identity matrices
\(\seq{I_{m_i}} \equiv \seq{I_{m_i} | 1 \leq i \leq k}\) and
\(\seq{I_{n_i}} \equiv \seq{I_{n_i} | 1 \leq i \leq k}\) on a standalone
basis. The sequences \(\seq{I_{m_i} \otimes \vb_i}\) and \(\seq{\a_i
\otimes I_{n_i}}\) can then be expressed as functions \(f: S \rightarrow
S\) and \(g: S \rightarrow S\) where \(S\) represents the set of all
matrix sequences of length \(k\). such that:

\begin{flalign*}
& f \brac{\seq{I_{m_i}}}
  = \seq{\seq{I_{m_i}}_i \otimes \seq{\vb}_i | 1 \leq i \leq k}
  \equiv \seq{I_{m_i} \otimes \vb_i | 1 \leq i \leq k}
  \equiv \seq{I_{m_i} \otimes \vb_i} & \\[10pt]
& g \brac{\seq{I_{n_i}}}
  = \seq{\seq{\a}_i \otimes \seq{I_{n_i}}_i | 1 \leq i \leq k}
  \equiv \seq{\a_i \otimes I_{n_i} | 1 \leq i \leq k}
  \equiv \seq{\a_i \otimes I_{n_i}} &
\end{flalign*}

Proposition \ref{prop:bds_vs_mmul} applies because we can express
\(\seq{\a_i \otimes I_{n_i}}\) and \(\seq{I_{m_i} \otimes \vb_i}\) as
left-multiplications over the sequences of identity matrcies
\(\seq{I_{m_i}}\) and \(\seq{I_{n_i}}\):

\begin{flalign*}
& \seq{I_{m_i} \otimes \vb_i} = \seq{\brac{I_{m_i} \otimes \vb_i} I_{m_i}} & \\[10pt]
& \seq{\a_i \otimes I_{n_i}} = \seq{\brac{\a_i \otimes I_{n_i}} I_{n_i}} &
\end{flalign*}

It then follows that:

\begin{flalign*}
& f \brac{\seq{I_{m_i}}}^{\oplus}_v
  = \left[ f \brac{\seq{\seq{I_{m_i}}^{\oplus}_v}} \right]_v & \\[10pt]
& g \brac{\seq{I_{n_i}}}^{\oplus}_v
  = \left[ g \brac{\seq{\seq{I_{n_i}}^{\oplus}_v}} \right]_v &
\end{flalign*}

Next, note that \(\seq{I_{m_i}}^{\oplus}_v\) produces an identity matrix
of size \(M\) with a row block structure of \(\a\), while
\(\seq{I_{n_i}}^{\oplus}_v\) produces an identity matrix of size \(N\)
with a row block structure of \(\vb\). In other words,
\(\seq{I_{m_i}}^{\oplus}_v = I_{\a}\) and \(\seq{I_{m_i}}^{\oplus}_v =
I_{\vb}\). This means:

\begin{flalign*}
& \left[ f \brac{\seq{\seq{I_{m_i}}^{\oplus}_v}} \right]_v
  = \left[ f \brac{\seq{I_{\a}}} \right]_v
  = \left[ \seq{ \seq{ I_{\a} }_i \otimes \vb_i } \right]_v & \\[10pt]
& \left[ g \brac{\seq{\seq{I_{n_i}}^{\oplus}_v}} \right]_v
  = \left[ g \brac{\seq{I_{\vb}}} \right]_v
  = \left[ \seq{ \a_i \otimes \seq{I_{\vb}}_i } \right]_v &
\end{flalign*}

Lastly, we note that the matrix representation of a pairwise Kronecker
product over two sequences of matrices is the same as a blockwise
Kronecker product if the sequences represent matrix blocks along the
same dimension. For example, \(\left[ \seq{ \a_i \otimes \seq{I_{\vb}}_i
} \right]_v = \a \odot I_{\vb}\) because \(\a_i\) and
\(\seq{I_{\vb}}_i\) are row blocks of \(\a\) and \(I_{\vb}\),
respectively.

Putting all the steps together we have:

\begin{align}
& \a \odot \vb = \seq{\a_i \otimes I_{n_i}}^{\oplus} \vb
               = \left[ \seq{\a_i \otimes \seq{\seq{I_{n_i}}^{\oplus}_v}_i} \right]_v \vb
               = \left[ \seq{ \a_i \otimes \seq{I_{\vb}}_i } \right]_v \vb
               = \brac{ \a \odot I_{\vb} }  \vb \label{eqn:bk_deriv_3a} \\[10pt]
& \a \odot \vb = \seq{I_{m_i} \otimes \vb_i}^{\oplus} \a
               = \left[ \seq{\seq{\seq{I_{m_i}}^{\oplus}_v }_i \otimes \vb_i} \right]_v \a
               = \left[ \seq{\seq{I_{\vb}}_i \otimes \vb } \right]_v \a
               = \brac{ I_{\a} \odot \vb } \a \label{eqn:bk_deriv_3b}
\end{align}
\end{proof}
\section{Conditions for the Redundancy of PCA Decomposition}
\label{pca_redundancy}

\begin{proposition} \label{prop:pca_redundancy}
Take a \(\brac{1 \times m}\) vector of observed factors \(X\) whose
covariance matrix has full rank. Define the PCA-based diffusion indices
obtained from \(X\) as \(\tX^{di} := X \O^{di}\), where \(\O^{di}\) is
an \(\brac{m \times m}\) matrix of principal component weights. Assume
WLOG that the ``true'' latent factor \(\tx^*\) that drives the dependent
can be expressed as a linear combination of the diffusion indices such
that \(\tx^* := \tX^{di} \o^*\) for some weight vector \(\o^*\). It then
follows from the properties of PCA that \(\tx^*\) can be expressed in
terms of \(X\) directly. \end{proposition}

\begin{proof}

By the properties of PCA, \(\O^{di}\) is a rotation matrix which means
it is invertible with \(\brac{\O^{di}}^{-1} = \brac{\O^{di}}'\). It then
follows that any valid solution for \(\o^*\) in the DI space can be
expressed as \(\o^* = \brac{\O^{di}}' \o^{**}\) for an equally valid
weight vector \(\o^{**}\) in the observed variable space. Formally:

\begin{equation*} 
\tx^* = \tX^{di} \o^* = X \O^{di} \o^* = X \O^{di}
        \brac{\O^{di}}' \o^{**} = X \o^{**}
\end{equation*}

The same result extends to the multivariate case in which the dependent
is a function of more than one latent predictor because the derivations
apply after replacing the vectors \(\o^*\) and \(\o^{**}\) with matrices
\(\O^*\) and \(\O^{**}\) where each column is a weight vector for one
latent factor \(\tx^*_j\).

\end{proof}
\section{LARX: Derivation of the Coefficient and Weight Vectors}
\label{derivation_w_o_p_g}
This section covers the derivations of the sample estimates \(\wh\),
\(\oh\), \(\ph\) and \(\bh\) for the coefficient vectors \(\w\), \(\o\),
\(\p\) and \(\b\), respectively. The customary ``hat'' superscripts are
omitted from here on out because they are implied everywhere. The
derivations for the Lagrange multipliers are presented in
\ref{derivation_ry_rl} and \ref{derivation_tlx_tlp}.

First, we note that the Lagrangian optimisation problem (\ref{eqn:clarx_lagr})
is convex, which means that the solution is obtained by setting various
partial derivatives to zero. Second, we note that the properties of the
Kronecker product and the block-wise Kronecker product allow us to
factorise \(\pw\) and \(\bo\) for compatibility with traditional matrix
calculus:

\begin{subequations}
\begin{align}
& \pw = \wip = \piw \label{eqn:factorisation_p_w} & \\
& \bo = \oib = \bio \label{eqn:factorisation_b_o} &
\end{align}
\end{subequations}

Here, \(I_a\) represents an identity matrix of size \(a\) (where \(a\)
is a scalar value), while \(I_{\a}\) represents an identity matrix with
the same size and row block structure as vector \(\a\) (vectors are
conventionally represented by bold letters). The property of the
Kronecker product relevant for (\ref{eqn:factorisation_p_w}) is well
established but reiterated for completeness in equations
(\ref{eqn:kron_factorisation_matrices}-\ref{eqn:kron_factorisation_vectors}).
The factorisation of the block-wise Kronecker product used for
(\ref{eqn:factorisation_b_o}) is proved in \ref{block_kron_factorisation}.
Taking \(\Sm_{BA}\) to denote the transpose of \(\Sm_{AB}\), the
solution for \(\b\) becomes:

\begin{align} \label{eqn:clarx_solution_beta}
& \frac{\partial}{\partial \, \b} \, \Lagr = 2 \left[
    \oi' \Sm_{X} \oib -
    \oi' \Sm_{XY} \w +
    \oi' \Sm_{XA} \pw \right] \nonumber \\[10pt]
& \b = \left[ \oi' \Sm_{X} \oi \right]^{-1}
       \left[ \oi' \Sm_{XY} \w - \oi' \Sm_{XA} \pw \right]
\end{align}

Similarly, recalling that \(\pw = \wip\), the solution for \(\p\) is:

\begin{align} \label{eqn:clarx_solution_phi}
& \frac{\partial}{\partial \, \p} \, \Lagr = 2 \left[
    \wi' \Sm_{A} \wip -
    \wi' \Sm_{AY} \w +
    \wi' \Sm_{AX} \bo \right] \nonumber \\[10pt]
& \p = \left[ \wi' \Sm_{A} \wi \right]^{-1}
       \left[ \wi' \Sm_{AY} \w - \wi' \Sm_{AX} \bo \right]
\end{align}

For the dependent weight vector \(\w\), we recall that \(\pw = \piw\).
The partial derivative of \(\Lagr\) with respect to \(\w\) is then given
by:

\begin{equation*}
\begin{split}
\frac{\partial}{\partial \, \w} \, \Lagr
    & = 2 \left[ \left( 1 + \ly \right) \Sm_{Y} +
            \pi' \Sm_{A} \pi -
            \pi' \Sm_{AY} -
            \Sm_{YA} \pi
        \right] \w \\
    & - 2 \left[
            \Sm_{YX} \bo -
            \pi' \Sm_{AX} \bo
        \right] + \l_l \1_n
\end{split}
\end{equation*}

Note that \(\frac{\partial}{\partial \, \w} \, \w' \Sm_{YA} \piw\)
resolves to \(\left[ \pi' \Sm_{AY} + \Sm_{YA} \pi \right] \w\) because
the quadratic form is not symmetric. Setting the partial derivative to
zero and expressing in terms of \(\w\), we get:

\begin{equation} \label{eqn:clarx_derivation_w}
\begin{split}
(1 + \ly ) \, \Sm_{Y} \w & =
        \left[
            \pi' \Sm_{AY} +
            \Sm_{YA} \pi -
            \pi' \Sm_{A} \pi \right] \, \w \\[10pt]
            & + \left[
            \Sm_{YX} -
            \pi' \Sm_{AX}
        \right] \bo - \frac{\l_l}{2} \1_n
\end{split}
\end{equation}

For ease of notation, define:

\begin{align*}
\underset{1 \times n}{\tv_1} = & \left[
    \pi' \Sm_{AY} +
    \Sm_{YA} \pi -
    \pi' \Sm_{A} \pi
\right] \w \\[10pt]
\underset{1 \times n}{\tv_2} = & \left[\Sm_{YX} - \pi' \Sm_{AX} \right] \bo
\end{align*}

Equation (\ref{eqn:clarx_derivation_w}) then becomes:

\begin{equation} \label{eqn:clarx_derivation_w_short}
(1 + \ly ) \, \Sm_{Y} \w = \tv_1 + \tv_2 - \frac{\l_l}{2} \1_n
\end{equation}

Setting \(\ry = (1 + \ly)\), \(\r_l = \frac{\l_l}{2}\) and
pre-multiplying both sides by \(\frac{1}{\ry} \, \Sm_{Y}^{-1}\) we
get:

\begin{equation} \label{eqn:clarx_solution_w}
\w = \frac{1}{\ry} \left[ \Sm_{Y}^{-1} \left( \tv_1 + \tv_2 \right)
   - \r_l \Sm_{Y}^{-1} \1_n \right]
\end{equation}

For the weight vector \(\o\), we recall that \(\bo = \bio\).
Furthermore, we can refactor the constraint terms for compatibility with
traditional matrix calculus using the properties of the block-wise direct
sum operator. For the portfolio constraints on \(\o_j\) we can apply
Proposition \ref{prop:bds_vs_vectors} to obtain:

\begin{equation*} 
\odt \1_{\o} = \left( \t1od \right)' \o
\end{equation*}

For the variance constraints on \(\tx_j\), we note three things. First
of all, because \(\u\) has the same length and row block structure as
\(\b\), the term \(\uo\) can be factorised in the same way as \(\bo\),
namely:

\begin{equation*} 
\uo = \ui \o = \oi \u
\end{equation*}

Second, by applying Proposition \ref{prop:bds_vs_transpose} we have:

\begin{equation*} 
\left[ \uo' \right]^{\oplus} = \left[ \uo^{\oplus} \right]'
\end{equation*}

Third, the operation \(\uo\) can be expressed as a left-multiplication
over the sequence of blocks in \(\o\), i.e., \(\uo = \ui \o = \left[
\seq{\brac{\u_j \otimes \seq{I_{\o}}_j} \o_j | 1 \leq j \leq K}
\right]_v\), which means that according to Proposition \ref{prop:bds_vs_mmul}
we have:

\begin{equation*} 
\uo^{\oplus} \equiv \left[ \ui \o \right]^{\oplus} = \ui \od
\end{equation*}

Putting these transformations together, we can rewrite the block-wise
quadratic form as:

\begin{equation*}
\begin{split}
\left[ \uo' \right]^{\oplus} & \Smd_X \uo
 = \left[ \uo^{\oplus} \right]' \Smd_X \uo
   = \left\{ \left[ \ui \o \right]^{\oplus} \right\}' \Smd_X \uo \\
 & = \left[ \ui \od \right]' \Smd_X \uo
   = \brac{ \od }' \ui' \Smd_x \ui \o
\end{split}
\end{equation*}

Applying these transformations, the partial derivative of
(\ref{eqn:clarx_lagr}) with respect to \(\o\) is:

\begin{equation*}
\frac{\partial}{\partial \, \o } \, \Lagr
     = 2 \bi' \left[
            \Sm_{X} \bio -
            \Sm_{XY} \w +
            \Sm_{XA} \pw
        \right] + 2 \M_2 \od \tl_x  + \t1od \tl_p
\end{equation*}

where \(\M_2 = \ui' \Smd_X \ui\). Furthermore, Propositions
\ref{prop:bds_vs_odot} and \ref{prop:block_kron_factorisation} prove that:

\begin{equation*}
\od \tl_x = \brac{\o \odot I_K} \tl_x = \brac{\tl_x \odot I_{\o}} \o
\end{equation*}

Re-arranging for \(\o\), we get:

\begin{equation} \label{eqn:clarx_solution_o}
\o = \left[ \bi' \Sm_{X} \bi + \M_2 \brac{\tl_x \odot I_{\o}} \right]^{-1}
    \left[ \tv_3 - \frac{1}{2} \t1od \tl_p \right] 
\end{equation}

with \(\tv_3 = \bi' \left[ \Sm_{XY} - \Sm_{XA} \pi \right] \w\).
\section{LARX: Derivation of the Lagrange Multipliers for the Dependent}
\label{derivation_ry_rl}
Recalling that \(\ry = (1 + \ly)\) and \(\r_l = \frac{\l_l}{2}\),
rewrite equation (\ref{eqn:clarx_derivation_w_short}) as:

\begin{flalign} \label{eqn:clarx_ry_1}
& \ry \, \Sm_{Y} \w = \tv_1 + \tv_2 - \r_l \1_n &
\end{flalign}

To solve for \(\r_l\), pre-multiply both sides of (\ref{eqn:clarx_ry_1}) by
\(\1_n'\) and re-arrange:

\begin{flalign*}
& \ry \, \1_n' \Sm_{Y} \w = \1_n' \left( \tv_1 + \tv_2 \right) - n \r_l & \\[10pt]
& n \r_l = \1_n' \left( \tv_1 + \tv_2 \right) - \ry \, \1_n' \Sm_{Y} \w &
\end{flalign*}

Dividing both sides by \(n\) produces:

\begin{flalign} \label{eqn:clarx_solution_rl}
& \r_l = \frac{1}{n} \1_n' \left( \tv_1 + \tv_2 \right) -  \frac{1}{n} \ry \, \1_n' \Sm_{Y} \w &
\end{flalign}

We can find an alternative solution for \(\r_l\) by pre-multiplying both
sides of (\ref{eqn:clarx_ry_1}) with \(\w'\) and recalling \(\w' \Sm_{Y} \w =
\s2_y\) and \(\w' \1_n = l_y\):

\begin{flalign*}
& \ry \s2_y = \w' \left( \tv_1 + \tv_2 \right) - \r_l l_y & \\[10pt]
& \r_l l_y = \w' \left( \tv_1 + \tv_2 \right) - \ry \s2_y &
\end{flalign*}

dividing both sides by \(l_y\) we get:

\begin{flalign} \label{eqn:clarx_solution_rl_alt}
& \r_l = \frac{1}{l_y} \w' \left( \tv_1 + \tv_2 \right) -
  \frac{\s2_y}{l_y}\ry &
\end{flalign}

This alternative solution is not very practical because it does not
allow for the case of \(l_y = 0\). However, we can use it in conjunction
with (\ref{eqn:clarx_solution_rl}) to eliminate \(\r_l\) and solve for
\(\r_y\):

\begin{flalign*}
& \frac{1}{n} \1_n' \left( \tv_1 + \tv_2 \right) -  \frac{1}{n} \ry \, \1_n' \Sm_{Y} \w
  = \frac{1}{l_y} \w' \left( \tv_1 + \tv_2 \right) - \frac{\s2_y}{l_y}\ry & \\[10pt]
& \frac{\s2_y}{l_y}\ry - \frac{1}{n} \ry \, \1_n' \Sm_{Y} \w
  = \frac{1}{l_y} \w' \left( \tv_1 + \tv_2 \right) - \frac{1}{n} \1_n' \left( \tv_1 + \tv_2 \right) & \\[10pt]
& \ry \frac{n \s2_y - l_y \1_n' \Sm_{Y} \w}{n l_y}
  = \frac{\left( n \w - l_y \1_n \right)' \left( \tv_1 + \tv_2 \right) }{n l_y} &
\end{flalign*}

Rearranging for \(\ry\) we get:

\begin{flalign} \label{eqn:clarx_solution_ry}
& \ry = \frac{ \left( n \w - l_y \1_n \right)' \left( \tv_1 + \tv_2 \right) }
             {n \s2_y - l_y \1_n' \Sm_{Y} \w} &
\end{flalign}
\section{LARX: Derivation of the Lagrange Multipliers for the Explanatory}
\label{derivation_tlx_tlp}
Start from the first-order condition for \(\o\). Expressing in terms of
\(\tl_p\), we get:

\begin{flalign} \label{eqn:deriv_tlpx_1}
& \t1od \tl_p = 2 \tv_3 - 2 \tv_4 - 2 \ui' \Smd_X \ui \od \tl_x &
\end{flalign}

Define the following shorthand notations for convenience:

\begin{flalign*}
& \underset{K \times K}{\V} = \diag(\ts2_x) = \begin{pmatrix}
    \s2_{x,1}, & 0, & \cdots & 0 \\
    0, & \s2_{x,2}, & \cdots & 0 \\
    \vdots & \vdots & \ddots & \vdots \\
    0, & 0, & \cdots & \s2_{x,K}
\end{pmatrix}, \underset{K \times K}{\P} = \diag(\lp) = \begin{pmatrix}
    l_{p,1}, & 0, & \cdots & 0 \\
    0, & l_{p,2}, & \cdots & 0 \\
    \vdots & \vdots & \ddots & \vdots \\
    0, & 0, & \cdots & l_{p,K}
\end{pmatrix}, & \\[10pt]
& \underset{K \times K}{\M_1} = \left( \t1od \right)' \t1od = \begin{pmatrix}
    m_1, & 0, & \cdots & 0 \\
    0, & m_2, & \cdots & 0 \\
    \vdots & \vdots & \ddots & \vdots \\
    0, & 0, & \cdots & m_K
\end{pmatrix}, \quad \underset{K \times K}{\M_2} = \ui' \Smd_X \ui &
\end{flalign*}

A system of two vector equations, each with \(K\) rows and \(K\)
unknowns, is produced by pre-multiplying the first-order condition for
\(\o\) with \(\brac{\od}'\) and \(\brac{\t1od}'\), respectively:

\begin{subnumcases}{\hspace{-15em}}
\P \tl_p = 2 \odt \left( \tv_3 - \tv_4 \right) - 2 \V \tl_x
    \label{eqn:deriv_tlpx_cond_1} \\[10pt]
\M_1 \tl_p = 2 \left( \t1od \right)' \left( \tv_3 - \tv_4 \right)
                 - 2 \left( \t1od \right)' \M_2 \od \tl_x
    \label{eqn:deriv_tlpx_cond_2}
\end{subnumcases}

Pre-multiplying \ref{eqn:deriv_tlpx_cond_1} and
\ref{eqn:deriv_tlpx_cond_2} by \(\P^{-1}\) and \(\M_1^{-1}\),
respectively, we get:

\begin{subnumcases}{\hspace{-12em}}
\tl_p = 2 \P^{-1} \odt \left( \tv_3 - \tv_4 \right) - 2 \P^{-1} \V \tl_x &
    \label{eqn:clarx_solution_tlp_a} \\[10pt]
\tl_p = 2 \M_1^{-1} \left( \t1od \right)' \left( \tv_3 - \tv_4 \right) -
                2 \M_1^{-1} \left( \t1od \right)' \M_2 \od \tl_x &
    \label{eqn:clarx_solution_tlp_b}
\end{subnumcases}

The solution for \(\tl_p\) given by \ref{eqn:clarx_solution_tlp_b} is
more practical because \ref{eqn:clarx_solution_tlp_a} does not allow for
the case of zero-sum weights (zeros on the diagonal of \(\P\) would
make it uninvertible).

The solution for \(\tl_x\) can be derived by setting the right-hand side
of \ref{eqn:clarx_solution_tlp_a} equal to the right-hand side of
\ref{eqn:clarx_solution_tlp_b}:

\begin{flalign*}
& 2 \P^{-1} \odt \left( \tv_3 - \tv_4 \right) - 2 \P^{-1} \V \tl_x =
  2 \M_1^{-1} \left( \t1od \right)' \left( \tv_3 - \tv_4 \right) -
  2 \M_1^{-1} \left( \t1od \right)' \M_2 \od \tl_x &
\end{flalign*}

We can pre-multiply both sides by \(\M_1 \P\) to avoid problems with
inverting \(\P\) in the presence of zero-sum weight constraints (note
that \(\M_1 \P = \P \M_1\) because both are diagonal):

\begin{flalign*}
& 2 \M_1 \odt \left( \tv_3 - \tv_4 \right) - 2 \M_1 \V \tl_x =
  2 \P \left( \t1od \right)' \left( \tv_3 - \tv_4 \right) - 2 \P \left( \t1od \right)' \M_2 \od \tl_x &
\end{flalign*}

Re-arranging for \(\tl_x\) we get:

\begin{equation} \label{eqn:clarx_solution_tlx}
\tl_x = \left[ \M_1 \V - \P \left( \t1od \right)' \M_2 \od \right]^{-1}
  \left( \od \M_1 - \t1od \P \right)' \left( \tv_3 - \tv_4 \right)
\end{equation}
\section{Out-of-Sample Forecasting Performance of the Alternative Forecasting Models}
\label{empirical_results_rev}
\begin{figure}[H]
    \centering  
    \includegraphics[width=\textwidth]{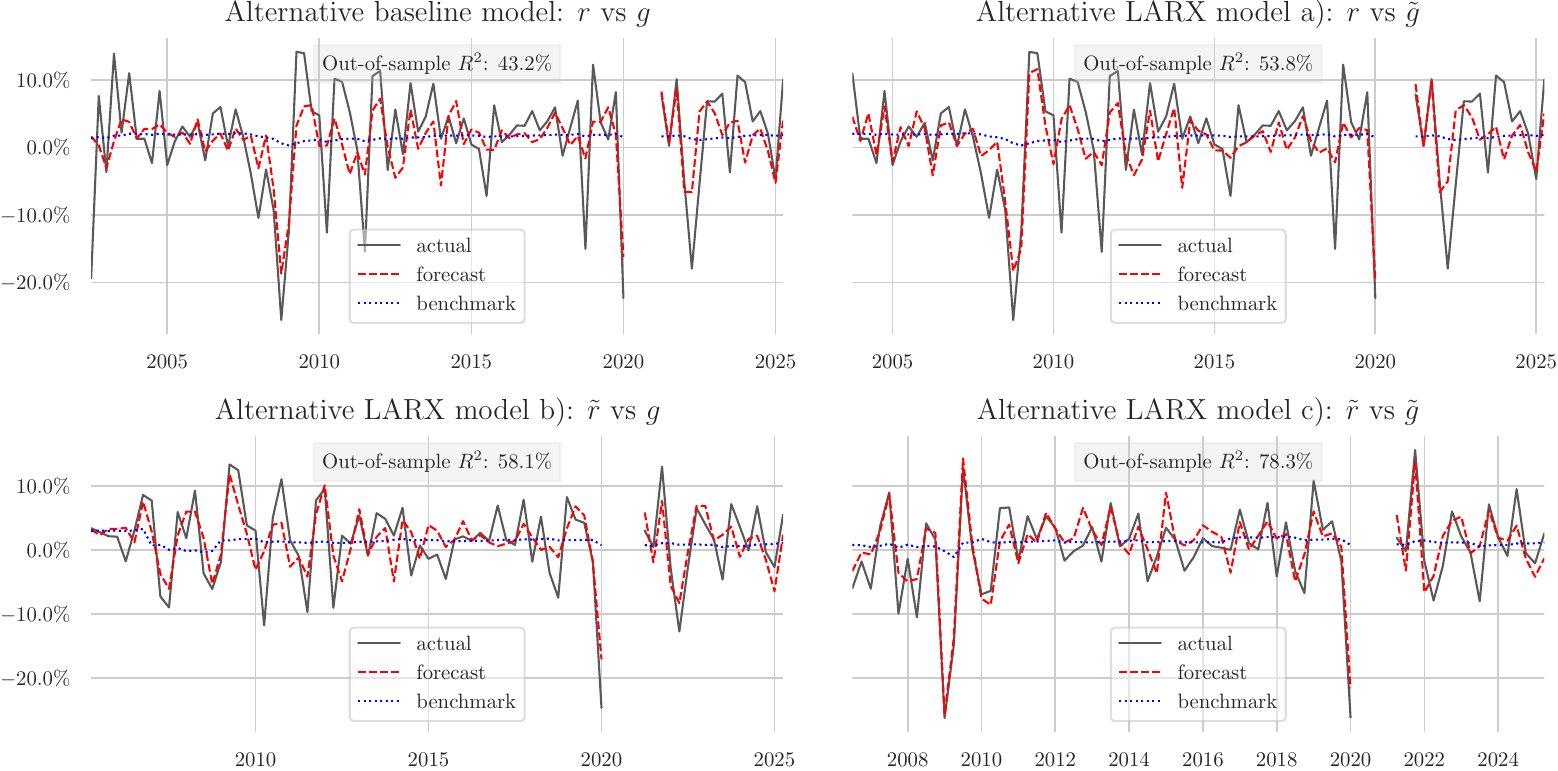}
    \caption{Out-of-sample performance: Market returns as the dependent} \label{NEW_fig-fc_ols_vs_clarx_reversed}
    \floatfoot{ \footnotesize %
      Note: Rolling out-of-sample forecasting performance of the
      alternative forecasting models with stock market returns as the
      dependent variable regressed on lags zero to 2 of real economic
      activity and 3 autoregressive lags as specified by equations
      (\ref{eqn:regression_models_rev})-(\ref{eqn:regression_all_latent_rev}).
      The target variance for the real activity SDI is set to the
      full-sample variance of quarterly US real GDP growth (approx.
      0.000124). The variance of the market SDI is set to the
      full-sample variance of the S\&P 500 quarterly log-return
      (approx. 0.006348).
    }
\end{figure}

\begin{subequations} \label{eqn:regression_models_rev}
Alternative baseline OLS/ARX model: The S\&P 500 \(r\) vs real GDP \(g\):
\begin{equation} \tag{\ref*{eqn:regression_models_rev}}
r_t = c + \sum_{\tau = 1}^3 \phi_{t-\tau} r_{t-\tau} +
          \sum_{\tau = 0}^2 \beta_{t-\tau} g_{t-\tau} + e
\end{equation}
Alternative LARX model: Latent explanatory: The S\&P 500 vs the real activity SDI
\(\tg\):
\begin{equation}\label{eqn:regression_r_nonlatent_rev}
r_t = c + \sum_{\tau = 1}^3 \phi_{t-\tau} r_{t-\tau} +
          \sum_{\tau = 0}^2 \beta_{t-\tau} \tg_{t-\tau} + e
\end{equation}
LARX model 2b: Latent dependent: Market SDI \(\tr\) vs real GDP:
\begin{equation}\label{eqn:regression_g_nonlatent_rev}
\tr_t = c + \sum_{\tau = 1}^3 \phi_{t-\tau} \tr_{t-\tau} +
            \sum_{\tau = 0}^2 \beta_{t-\tau} g_{t-\tau} + e
\end{equation}   
LARX model 2c: Both latent: Market SDI vs the real activity SDI:
\begin{equation}\label{eqn:regression_all_latent_rev}
\tr_t = c + \sum_{\tau = 1}^3 \phi_{t-\tau} \tr_{t-\tau} +
            \sum_{\tau = 0}^2 \beta_{t-\tau} \tg_{t-\tau} + e
\end{equation}
\end{subequations}
\section{Data reference}
\label{sec:org483dd8c}
\begin{table}[H]
\centering
\caption{Data series used in the empirical study}
\label{table:data_reference}
\begin{adjustbox}{width=\textwidth, nofloat=table}
\begin{threeparttable}
  \begin{tabular}{lllll}
  \toprule
  Dataset & Source & Ticker\tnote{1} & Frequency & History start \\
  \midrule
  Real GDP & U.S. Bureau of Economic Analysis & GDPC1 & Quarterly & 1947Q1 \\
  \quad Personal Consumption Expenditure (Cons.) & U.S. Bureau of Economic Analysis & PCECC96 & Quarterly & 1947Q1 \\
  \quad Gross Private Domestic Investment (Inv.) & U.S. Bureau of Economic Analysis & GPDIC1 & Quarterly & 1947Q1 \\
  \quad Government Consumption and Investment (Govt.) & U.S. Bureau of Economic Analysis & GCEC1 & Quarterly & 1947Q1 \\
  \quad Exports of Goods and Services (Exp.) & U.S. Bureau of Economic Analysis & EXPGSC1 & Quarterly & 1947Q1 \\
  \quad Imports of Goods and Services (Imp.) & U.S. Bureau of Economic Analysis & IMPGSC1 & Quarterly & 1947Q1 \\
  \hdashline S\&P 500 & Investing.com & US500 & Monthly & 1989-10 \\
  \quad Energy (En.) & Investing.com & SPNY & Monthly & 1989-10 \\
  \quad Materials (Mat.) & Investing.com & SPLRCM & Monthly & 1989-10 \\
  \quad Industrials (Ind.) & Investing.com & SPLRCI & Monthly & 1989-10 \\
  \quad Financials (Fin.) & Investing.com & SPSY & Monthly & 1989-10 \\
  \quad Healthcare (HC) & Investing.com & SPXHC & Monthly & 1989-10 \\
  \quad Consumer Discretionary (Disc.) & Investing.com & SPLRCD & Monthly & 1989-10 \\
  \quad Consumer Staples (Stapl.) & Investing.com & SPLRCS & Monthly & 1989-10 \\
  \quad Communication (Telco) & Investing.com & SPLRCL & Monthly & 1989-10 \\
  \quad Technology (IT) & Investing.com & SPLRCT & Monthly & 1989-10 \\
  \quad Utilities (Util.) & Investing.com & SPLRCU & Monthly & 1989-10 \\
  \bottomrule
  \end{tabular}
  
\begin{tablenotes}[flushleft]
\item[1] Data for U.S. GDP and its individual expenditure components was
retrieved from the St. Louis Federal Reserve economic database (FRED)
on 20 October 2025. The corresponding tickers are identifiers for the
FRED database. Data for the S\&P 500 and its sector sub-indices was
retrieved directly from Investing.com on 20 October 2025 using the
tickers above.
\end{tablenotes}
\end{threeparttable}
\end{adjustbox}
\end{table}

\end{appendix}

\bibliographystyle{unsrtnat}
\bibliography{/home/daniil/Research/bibliography/phd1}
\end{document}